\documentclass[onefignum,onetabnum]{siamonline220329_XCharter}
\usepackage[margin=1.2in]{geometry}
\usepackage[sups,osf]{XCharter}
\usepackage[utf8]{inputenc} 
\usepackage[T1]{fontenc}
\usepackage{makeidx}         
\usepackage{graphicx}        
\usepackage{multicol}        
\usepackage[bottom]{footmisc}

\usepackage{hyperref}
\usepackage{dsfont}
\usepackage[varvw]{newtxmath}       

\makeindex

\usepackage{cleveref}

\newcommand{\Eb}{{\mathbb{E}}}
\newcommand{\R}{{\mathbb{R}}}
\newcommand{\N}{{\mathbb{N}}}

\DeclareMathOperator{\argmin}{argmin}
\DeclareMathOperator{\conv}{conv}
\newcommand{\diag}{{\rm{diag}}}

\newtheorem{remark}[theorem]{Remark}

\crefformat{equation}{(#2#1#3)}
\crefmultiformat{equation}{(#2#1#3)}
{ and~(#2#1#3)}{, (#2#1#3)}{ and~(#2#1#3)}

\title{Sparse Recovery for Overcomplete Frames: Sensing Matrices and Recovery Guarantees}
\author{Xuemei Chen\thanks{Department of Mathematics and Statistics, University of North Carolina Wilmington, 601 S. College Rd., Wilmington, NC 28403 (\email{chenxuemei@uncw.edu}).} \and  Christian K\"ummerle\thanks{Department of Computer Science, The University of North Carolina at Charlotte,
9201 University City Blvd, Charlotte, NC 28223 \email{(kuemmerle@charlotte.edu)}.}
 \and Rongrong Wang\thanks{Department of Computational Mathematics, Science and Engineering and Department of Mathematics, Michigan State University, 426 Auditorium Rd.,
East Lansing, MI 48824 (\email{wangron6@msu.edu}).}}

\begin{document}
\maketitle

\begin{abstract}
Signal models formed as  linear combinations of few atoms from an over-complete dictionary or few frame vectors from a redundant frame have become central to many applications in high dimensional signal processing and data analysis. A core question is, by exploiting the intrinsic low dimensional structure of the signal, how to design the sensing process and decoder in a way that the number of measurements is essentially close to the complexity of the signal set. 
This chapter provides a survey of important results in answering this question, with an emphasis on a basis pursuit like convex optimization decoder that admits a wide range of random sensing matrices. The results are quite established in the case signals are sparse in an orthonormal basis, while the case with frame sparse signals is much less explored. In addition to presenting the latest results on recovery guarantee and how few random heavier-tailed measurements fulfill these recovery guarantees, this chapter also aims to provide some insights in proof techniques.
We also take the opportunity of this book chapter to publish an interesting result (\Cref{thm:dripbad}) about a restricted isometry like property related to a frame.
\end{abstract}

\section{Introduction}\label{sec:intro}

In the last two decades, the problem of recovering structured signals from significantly under-sampled linear measurements, often referred to as \emph{compressed sensing} or \emph{compressive sensing}, has been extensively studied by applied mathematicians, signal processing researchers,  information theorists and statisticians since building on and extending the seminal works of Donoho, Cand\`es, Romberg and Tao \cite{CRT06, CT06, D06}, which showed that a generic variant of the problem can be solved efficiently from a number of random linear measurements that is essentially linear in the dimension of the set of structured signals. The ensued line of research has led to the popularization of related data-efficient, optimization-based decoders in seismology and geophysics \cite{Log65,Levy-1981Reconstruction,Santosa86,Donoho-1992Signal,Her10,Baraniuk-2017Compressive}, magnetic resonance tomography and medical imaging \cite{LDP07,KrahmerWard-2014,Jaspan-2015compressed,FWHS16,Don18,Ravishankar2019image,Adcock-2021Compressive}, and machine learning \cite{Tibshirani-1996Lasso,EladAharon-2006ImageDenoising,Mairal-2008Supervised,Hastie-2015statistical,Papyan-2018Theoretical}. 

Formally, in this problem, we wish to recover $z_0\in\R^d$ from its under-sampled and possibly corrupted linear measurement $y=Az_0+w\in\R^m$ where $m$ is typically much less than $d$ and $w$ represents the measurement corruption. Given that $A\in\R^{m\times d}$ is underdetermined, such recovery is only possible if we have more information about $z_0$. Additionally, due to the unavoidable presence of noise $w$, we are interested in a recovery procedure which is able to identify $z_0$ up to an error that is proportional to a norm $\|w\|$ of $w$.

While this problem is ill-posed without further assumption, in the general setting of compressed sensing, the problem is made well-posed by assuming that $z_0$ is an element of a set with an underlying low dimensional structure. Specifically, $z_0$ could be sparse in the standard basis (which means that the majority of its coordinates are zero), in an application-tailored orthogonal basis such as a suitable wavelet \cite{LDP07,Mallat-1999Wavelet} or shearlet basis \cite{Kutyniok-2012Sompactly}, or be written as a linear combination of only a few atoms in a redundant dictionary \cite{Elad07,DMP1,CENR11, MBKW23}. Similarly, $z_0$ can be a matrix of low rank, in which case the $A$ corresponds to a linear operator acting on matrices. In this case, the problem is known under the name \emph{low-rank matrix recovery} \cite{Davenport16,CaiWeiExploiting18,chen_chi18,ChiLuChen19}, and has applications in recommender systems \cite{koren_bell_volinsky,recht2013parallel}, image processing \cite{Jin-2015Annihilating,Sengupta-2017New},  computational physics \cite{Shechtman-2015Phase,Candes-2015Phase} and control theory \cite{Fazel_hindi_boyd04,Dorfler-2022Bridging}.

The theoretical understanding of the compressed sensing problem and of the properties of suitable, efficient reconstruction algorithms has reached some maturity in the literature. We point to the monographs \cite{FR13,Rish-2014Sparse,Zhao-2018Sparse,Adcock-2021Compressive} and references therein for an overview of known results. Conditions for low-rank matrix recovery using tractable algorithms have similarly been established in the last few years; for that problem, there are fewer comprehensive works covering the general theory, but \cite{Gross-2011Recovering,Chen-2015Incoherence,Davenport16,CaiWeiExploiting18,ChiLuChen19} contain the majority of the most relevant results.

However, the vast majority of results for these problems study recovery guarantees if the signal $z_0$ is sparse or low-rank with respect to the standard basis or an orthogonal basis set. In many application domains, this assumption is too rigid to be satisfied, or better reconstruction can be achieved by considering sparsity of the signal with respect with to a transform domain, which is typically chosen using a certain degree of redundancy~\cite{EladAharon-2006ImageDenoising, Hastie-2015statistical, Mairal-2008Supervised, CCS08}. 
Specifically, in this chapter, we present and review results on theoretical guarantees for parsimonious recovery problems under this more realistic modeling.

In particular, we will focus on the case where $z_0$ is sparse or approximately sparse in a dictionary or frame $F=\{f_1,f_2,\cdots, f_n\}$.\footnote{A set of vector $F=\{f_i\}_{i\in I}$ is a \emph{frame} for a Hilbert space $H$ if there exists $0<A\leq B<\infty$ such that $A\|x\|^2\leq\sum_{i\in I}|\langle x, f_i\rangle|^2\leq B\|x\|^2$ for all $x\in H$.}  In the finite dimensional setting we are considering, we can think of $F$ as a spanning set of $\R^d$. Moreover, when appropriate, we will also use $F$ for the matrix $F = [f_1, f_2,\cdots, f_n] \in \R^{d \times n}$ collecting the frame elements of the frame. With this preparation, we can define, for a positive integer $k$ representing a sparsity level, the set of \emph{$F$-$k$-sparse} vectors $z \in \R^d$ as
\begin{equation} \label{eq:SigmaFk}
\Sigma_{F,k}:=\Big\{z\in\R^d: z=\sum_{i=1}^n x_if_i=Fx, \|x\|_0\leq k\Big\}
\end{equation}
with $\|x\|_0$ denoting the number of non-zero entries of $x$.

For the compressed sensing problem of recovering vectors sparse with respect to a frame, the key questions  have been focused on:
\begin{itemize}
\item[Q1.]\label{Q1} \ \ Which tractable algorithm or decoder can be used to recover $z_0\in\Sigma_{F,k}$, given access to $A$ and $y=Az_0$?
\item[Q2.]\label{Q2} \ \ Given $F$, how to design the sensing matrix $A$ so that the decoder found in Q1 is successful?
\item[Q3.] \ \ For the design in Q2, what is the minimum number of measurements $m$?
\item[Q4.]\label{Q4} \ \ With the appropriately designed $A$, is the decoder found in Q1 stable and robust with respect to model misfit and additive noise?
\end{itemize}

In this chapter, we provide a short survey about relevant results addressing these questions.
To quantify the model misfit with regards to the signal space \cref{eq:SigmaFk}, we define
\begin{equation}\label{equ:tail}
\sigma_{F, k}(z):=\min_{v\in\Sigma_{F,k}}\|z-v\|_F,
\end{equation}
the distance of $z$ to the set of $F$-$k$-sparse vectors. Here the $F$-norm is defined using the $\ell_1$-norm 
\begin{equation}\label{equ:zF}
\|z\|_{F}:=\min\{\|x\|_1 : Fx=z\},
\end{equation}
based on coefficients relative to the frame $F$.

To be more specific about Q4, we would like the decoder $\Delta : \R^m \to \R^d$ to perform relatively well even when the signals are not exactly $F$-$k$-sparse (meaning that $\sigma_{F, k}(z)>0$; such a property is known as \emph{stability} \cite[Section 4.2]{FR13}), and also under the presence of measurement noise $w$ (this is known as \emph{robustness} \cite[Section 4.3]{FR13}). Quantifying the latter using a suitable norm $\|w \|$ on $w$, it is desirable to obtain a stable and robust recovery guarantee such that there exist constants $C_1, C_2$ depending on $A$ and possibly the dimensions which allow the inequality
\begin{equation}\label{equ:performance}
\|z - \Delta(A z + w)\|\leq C_1\sigma_{F,k}(z)+C_2\|w\|,
\end{equation}
to hold for all $z \in \R^d$.

Regarding Q1, as common in the compressed sensing literature, we distinguish two types of efficient decoders: iterative solvers such as greedy methods \cite{Giryes-2015Greedy,Tirer-2020GeneralizingCoSaMP} or thresholding based methods \cite{PelegElad-2013} (among which compressive sampling matching pursuit \cite{DNW,Giryes-2014Greedy} and iterative hard thresholding \cite{Foucart-2016Dictionary} are some of the most well-known examples) and decoders based on the solution of an optimization problem. 

Due to their popularity in imaging problems \cite{Adcock-2021Compressive} and the maturity of their theory, we will focus on optimization-based decoders defined as the solution of optimization problems of the form
\begin{equation}\label{equ:f}
\Delta_{f, A, \eta}(y):=\underset{z \in \R^d}{\argmin}~ f(z) \quad\text{ subject to }\quad \|Az - y\|_2\leq\eta,
\end{equation}
where $f: \R^d \to \R$ is a suitable objective function \cite{Decode,CDD09,FR13,ACP11,MBKW23}. It is understood that the objective function $f$ needs to be selected in a way that is simultaneously sparsity-promoting, can be evaluated efficiently, and also such that the respective optimization problem of \cref{equ:f} is easy or at least tractable to solve. While non-convex choices of $f$ \cite{Foucart-2009Lq, Saab-2010Sparse,Wipf-2010Iterative,ACP11,Zheng2017-LpL1}, \cite[Chapter 7]{Zhao-2018Sparse} often lead to improved performance compared to convex sparsity-promoting functions, choosing $f$ as a convex function has the advantage that properties of the solution of \cref{equ:f} can be decoupled from the particular solver to be used and that they can be analyzed using mature tools from convex optimization \cite{FR13,Adcock-2021Compressive}.

For the  frame-based recovery problem \cref{eq:SigmaFk},  the convex $F$-norm $f(z) = \|z\|_{F}$ of \cref{equ:zF} happens to be a suitable objective function, in which case the decoder \cref{equ:f} is denoted as
\begin{equation}\label{equ:F}
\Delta_{F, A, \eta}(y)=\underset{z \in \R^d}{\argmin}~ \|z\|_F \quad\text{ subject to }\quad \|Az - y\|_2\leq\eta.
\end{equation}

\paragraph*{Notations and Organization} For an integer $N$, we let $[N]$ be the index set $\{1, 2, \cdots, N\}$. For a vector $x$, $\|x\|_p=(\sum |x_i|^p)^{1/p}$ is its $p$-norm. If $T$ is an index set, then $x_T$ is the vector that has the same value as $x$ on $T$ and $0$ elsewhere.
For a matrix $A$,  $\ker(A)$ is the kernel (null space) of $A$, and $\|A\|_2$ is its spectral norm. $I_N$ is the $N\times N$ identity matrix.

The rest of the chapter is organized as follows. Section \ref{sec:basis:guarantee} provides recovery guarantees of the $\ell_1$ minimization problem for recovering vectors sparse in an orthonormal basis under various conditions, as listed in Definition \ref{def:basisp}, on the sensing matrix $A$. We make a comparison among these conditions towards the end. Section \ref{sec:recoveryguarantees:frames} explores recovery guarantees of \eqref{equ:F}, i.e., the $\ell_1$ synthesis method with the general frame $F$. Our attention focuses on the null space property like conditions for the sensing matrix $A$. The $\ell_1$ analysis method is also touched upon as a related model. Section \ref{sec:random} focuses on how random sensing matrices fulfill the conditions laid out in \Cref{def:basisp} or \Cref{def:Fp}. Sub-Gaussian matrices have been typical in this regard while we also provide a wider range of random measurements with heavier-tail behavior than sub-Gaussian.

\section{Recovery Guarantees for Signals with Sparsity in Orthonormal Basis}\label{sec:basis:guarantee}
We first present relevant results for the case that the frame matrix $F$ is the identity matrix, which corresponds to the setup of the initial works  of compressed sensing theory \cite{CRT06, CT06, D06,CDD09,FR13}. In this case, \eqref{equ:F} becomes
\begin{equation}\label{equ:bp}
\Delta_{1,A,\eta}(y):=\Delta_{I_d,A,\eta}(y)=\underset{z \in \R^d: \|Az - y\|_2\leq\eta }{\argmin} \|z\|_1,
\end{equation}
which is also known by the name of \emph{quadratically constrained basis pursuit} \cite{CRT06,Brugiapaglia-2018Robustness,Adcock-2021Compressive}, and for $\eta=0$, it can be simplified to the equality-constrained \emph{basis pursuit} \cite{Chen-2001Atomic,FR13} decoder
\begin{equation}\label{equ:bp0}
\Delta_{1,A,0}(y) =\argmin_z \|z\|_1 \quad\text{ subject to } \quad Az =y.
\end{equation}
Equality-constrained \emph{basis pursuit} is a suitable decoder in the case of a non-existent or small noise level, or in the case of unknown noise level \cite{W10, F14,Brugiapaglia-2018Robustness}.

Suitable recovery guarantees for \eqref{equ:bp} and \eqref{equ:bp0} can in most cases easily be extended to $\Delta_{F, A, \eta}(y)$ or $\Delta_{F,A,0}(y)$ and a signal model in which $F$ is not the identity, but an orthonormal basis, by incorporating $F$ into the measurement matrix $A$ as  $\Delta_{I_d, A F, \eta}(y)$. We shorten the notation to $\Sigma_k=\Sigma_{I_d, k}$ and $\sigma_k(z)=\sigma_{I_d, k}(z)$ for this section.

To answer \hyperref[Q4]{Q2}, a series of properties on the sensing matrix have been proposed over the years, the most relevant of which we list below. These properties can be used to characterize the performance of the decoders $\Delta_{1,A,\eta}$ for problems involving such sensing matrices.
\begin{definition}\label{def:basisp}
Let $A\in\R^{m\times d}$.
\begin{enumerate}
\item $A$ is said to have the \emph{null space property} of order $k$~\cite{DE03, GN03, CDD09} if 
\begin{equation}\tag{NSP-$k$}\label{nsp}
\|z_T\|_1<\|z_{T^c}\|_1, \quad \forall z\in\ker(A)\backslash\{0\}, \forall T \text{ with } |T|\leq k.
\end{equation}
\item $A$ is said to satisfy the \emph{restricted isometry property}~\cite{CT06} of order $k$ with constant $0 \leq \delta < 1$  if 
\begin{equation} \tag{{RIP-(k,\ensuremath{\delta})}} \label{rip}
(1 - \delta) \|z\|_2^2 \leq \|Az\|^2_2 \leq (1 + \delta) \|z\|_2^2, \quad \forall  z\in\Sigma_k.
\end{equation}
\item If $A$ has $\ell_2$-normalized columns $a_1, a_2, \cdots, a_d$, then the \emph{coherence} of $A$~\cite{T04} is defined as
\begin{equation} \label{eq:def:coherence}	
\mu(A):=\max_{i\neq j}|\langle a_i, a_j\rangle|.
\end{equation}
\item $A$ is said to have the \emph{$\ell_q$-robust null space property} of order $k$ with $\rho>0, \tau>0$~\cite{F14} if
\begin{equation}\tag{RNSP-\ensuremath{(q,k,\rho,\tau)}} \label{rnsp}
\|z_T\|_q\leq\frac{\rho}{k^{1-1/q}}\|z_{T^c}\|_1+\tau\|Az\|_2, \forall z \in \R^{d} , \forall T \text{ with } |T|\leq k.
\end{equation}
\item $A$ is said to have the \emph{quotient property}~\cite{W10, F14} with constant $\alpha>0$ if 
\begin{equation}\tag{QP-\ensuremath{\alpha}} \label{qp}
\|y\|_A\leq\alpha \sqrt{k_*}\|y\|_2, \forall y\in\R^m,
\end{equation}
where $k^*=m/\log(ed/m)$. 
\item $A$ is said to satisfy the \emph{robust width property} of order $k$ with constants $c_0, c_1$~\cite{CM21} if
\begin{equation}\tag{{RWP-\ensuremath{(k,c_0,c_1)}}} \label{rwp}
\|z\|_2\leq\frac{c_0}{\sqrt{k}}\|z\|_1, \forall z \in \R^d{} \text{ such that }\|Az\|_2<c_1\|z\|_2.
\end{equation}
\end{enumerate}
\end{definition}

Given the definition \eqref{equ:zF}, the quotient property with constant $\alpha$ \eqref{qp} holds if and only if for any $y\in\R^m, \text{ there exists }z\in\R^d\text{ such that }Az=y\text{ and }\|z\|_1\leq\alpha \sqrt{k_*}\|y\|_2
$. \eqref{qp} can be also interpreted as a geometric property \cite{Gluskin-1989Extremal,Guedon-2022Geometry} stating that the $\ell_2$-ball $B_2^d$ is contained in an appropriately scaled centrally symmetric polytope spanend by the columns of $A = \begin{bmatrix}
	a_1, & a_2, & \ldots, & a_d
\end{bmatrix}$, which can be written as $B_2^d\subseteq \alpha\sqrt{k_*}A(B_1^m)$ if $A(B_1^m)$ is the image of the $\ell_1$-ball $B_1^m$ with respect to $A$. 
The \emph{width property}, proposed in  \cite{KT07}, is a weaker version of \eqref{rwp} and closely related to the null space property \eqref{nsp}.

We sample some important results affirmatively answering \hyperref[Q4]{Q4} for suitable parameter choices.
\begin{theorem}[\cite{CZ13}]\label{thm:rip}
If $A$ satisfies the RIP-$(2k,\delta)$ with $\delta<1/\sqrt{2}$, then there exists constants $D_1,D_2 > 0$ such that
\begin{equation}\label{equ:thmrip}
\|z-\Delta_{1,A,\eta}(Az+w)\|_2\leq D_1\frac{\sigma_k(z)}{\sqrt{k}}+D_2\eta
\end{equation}
for all $z\in\R^d$ and $w\in\R^m$ with $\|w\|_2\leq\eta$. Both $D_1, D_2$ only depend on $\delta$.
\end{theorem}
Inequality \eqref{equ:thmrip} provided by \Cref{thm:rip} provides a stable and robust recovery guarantee for the decoder $\Delta_{1,A,\eta}$ as outlined in  \eqref{equ:performance} in the case that $A$ satisfies an RIP of order $2k$ and if the noise $\ell_2$-norm satisfies $\|w\|_2\leq\eta$.

\begin{theorem}[{\cite{DE03, GN03}}]\label{thm:nsp}
$\Delta_{1,A,0}(Az)=z$ for any $z\in\Sigma_k$ if and only if $A$ satisfies \eqref{nsp}.
\end{theorem}

\Cref{thm:nsp} provides a sharp characterization of the success of the basis pursuit decoder $\Delta_{1,A,0}$ in absence of additive noise $w$ for exactly $k$-sparse coefficient vectors $z$.
\begin{theorem}[\cite{GN03}] \label{thm:coherence}
If $A$ has $\ell_2$-normalized columns and its coherence $\mu(A)$ (see \eqref{eq:def:coherence}) satisfies $\mu(A)<\frac{1}{2k-1}$, then $\Delta_{1,A,0}(Az)=z$ for any $z\in\Sigma_k$.
\end{theorem}
While the coherence assumption of \Cref{thm:coherence} is strictly stronger than the NSP assumption \eqref{nsp}, the latter is NP-hard to certify computationally given a matrix $A$ \cite{Tillmann-2013Computational}, whereas $\mu(A)$ can be computed easily.

With the generalizations \eqref{rnsp} of the null space property \eqref{nsp}, it is possible to obtain robust recovery guarantees of the type \eqref{equ:performance} as well:
\begin{theorem}[{\cite[Theorem 5]{F14}}] \label{theorem:RNSP:guarantee}
Given $q\geq1$, if $A\in\R^{m\times d}$ satisfies \eqref{rnsp} with $0<\rho<1$, then for all $z\in\R^d$ and $w\in\R^m$ with $\|w\|_2\leq\eta$,
\begin{equation}\label{equ:thm:rnsp:1}
	\|z-\Delta_{1,A,\eta}(Az+w)\|_1\leq\frac{2(1+\rho)}{1-\rho}\sigma_k(z)+\frac{2\tau}{1-\rho}k^{1-\frac{1}{q}}\eta,
\end{equation}
and \begin{equation}\label{equ:thm:rnsp:q}
	\|z-\Delta_{1,A,\eta}(Az+w)\|_p\leq\frac{2(1+\rho)^2}{1-\rho}\frac{\sigma_k(z)}{k^{1-1/p}}+\frac{3+\rho}{1-\rho} \tau k^{\frac{1}{p}-\frac{1}{q}}\eta
\end{equation}
for $1 < p \leq q$. 
\end{theorem}
Similar results can be obtained via robust width properties \eqref{rwp} of $A$. 
\begin{theorem}[{\cite[Theorem 3]{CM21}}]\label{thm:rwp}
If $A$ satisfies \eqref{rwp}, then
\begin{equation}\label{equ:thmrwp}
\|z-\Delta_{1,A,\eta}(Az+w)\|_2\leq 4c_0\frac{\sigma_k(z)}{\sqrt{k}}+\frac{2}{c_1}\eta.
\end{equation}
for all $z\in\R^d$ and $w\in\R^m$ with $\|w\|_2\leq\eta$. 
\end{theorem}

Unlike the other properties defined in \Cref{def:basisp}, the quotient property \eqref{qp} alone does not provide any recovery guarantees for compressed sensing decoders. However, if the sensing matrix satisfies \eqref{qp} in addition to another property which guarantees recovery under noiseless measurements, it can be shown that equality-constrained basis pursuit $\Delta_{1,A,0}$ of \eqref{equ:bp0} is likewise a noise-robust decoder, see \cite{W10}, \cite[Section 11.2]{F14}. 

More specifically, $y= Az_0+w$ can be seen as a perturbed version of the true measurement. We note that the programs $\Delta_{1,A,\eta}$ or $\Delta_{f, A, \eta}$ from \eqref{equ:bp} and  \eqref{equ:f} need the noise strength level $\eta$ as an input parameter, which may not be  accessible in many applications. The advantage of $\Delta_{1,A,0}$ or $\Delta_{f,A,0}$ is that they amount to a noise-blind method, where robust reconstruction can still be obtained with reasonable sensing matrices.
We will cite a result from \cite{F14}.
\begin{theorem}[{\cite[Theorem 11]{F14}}]\label{thm:equality}
If $A\in\R^{m\times d}$ satisfies RNSP($q, k,\rho,\tau k_*^{1/q-1/2}$) with $k\leq ck_*$ and \eqref{qp}, then for any $1\leq p\leq q$,
\begin{equation}
\|z-\Delta_{1,A,0}(Az+w)\|_p\leq C\frac{\sigma_k(z)}{k^{1-1/p}}+Dk_*^{\frac{1}{p}-\frac{1}{2}}\|w\|_2,
\end{equation}
for all $z\in\R^d$ and $w\in\R^m$. The constants $C, D>0$ depend only on $\rho, \tau, \alpha, c$.
\end{theorem}

The restricted isometry property \eqref{rip} is considered as the strongest one among the sparse recovery inducing properties of \Cref{def:basisp}. They are related to each other as follows.
\begin{itemize}
\item RIP implies RNSP. It is stated in \cite[Theorem 6.13]{FR13} that if $A$ satisfies RIP-$(2k,\delta)$, then $A$ satisfies the RNSP-$(2,k,\rho,\tau)$ with $\rho = \frac{\delta}{\sqrt{1-\delta^2} - \delta/4}$ and $\tau = \frac{\sqrt{1+\delta}}{\sqrt{1-\delta^2}-\delta/4}$. If  $\delta < \frac{4}{\sqrt{41}} \approx 0.6246$, then it holds that $\rho < 1$, which is necessary for making the robust guarantees \eqref{equ:thm:rnsp:q} work.\footnote{We note that the stable and robust guarantee \eqref{equ:thm:rnsp:q} with $p=2$ for $\Delta_{1,A,\eta}$ is implied by the RIP-$(2k,\delta)$ even for the weaker (and optimal) assumption $\delta < \sqrt{2}/2 \approx 0.7071$ on $\delta$, as shown in \cite[Theorem 2.1]{Cai-2013Sparse}.} 
\item RNSP implies RWP. Theorem \ref{thm:rwp} also goes the other direction. That is, if for all $z\in\R^d$ and $w\in\R^m$ with $\|w\|_2\leq\eta$, it holds $\|z-\Delta_{1,A,\eta}(Az+w)\|_2\leq C_1\frac{\sigma_k(z)}{\sqrt{k}}+C_2\eta$, then $A$ must satisfy RWP-$(k, 2C_0, \frac{1}{2C_1})$. Therefore with Theorem \ref{theorem:RNSP:guarantee}, we see RNSP-$(2,k,\rho,\tau)$ implies RWP of order $k$ with constants derived from \eqref{equ:thm:rnsp:q}.
\item \eqref{nsp} is the weakest of all due to Theorem \ref{thm:nsp}. Moreover, only the NSP is a property which solely depends on the kernel of $A$. In order to have a fair comparison between the reconstruction performance of, say RIP, and NSP, we need to consider sensing matrices that have the same null space as some RIP matrix. See~\cite{CCW16} for more details.
\item QP does not imply RIP, nor vice versa \cite{W10}.
\end{itemize}

NSP is sensitive to column scaling, as shown below.

\begin{lemma}\label{lem:nsp}
If $B$ satisfies \eqref{nsp}, then there exists a diagonal matrix $D$ such that $BD$ does not satisfy \eqref{nsp}.
\end{lemma}
\begin{proof}
Fix $v\in\ker(B)\backslash\{0\}$ and choose $|T|=k$, such that $0<\|v_T\|_1<\|v_{T^c}\|_1$. 
Let $(\diag(D))_T = \frac{\|v_T\|_1}{\|v_{T^c}\|_1+1}$ and $(\diag(D))_{T^c} =1$.

Define $w$ such that $w_T=\frac{\|v_{T^c}\|+1}{\|v_T\|_1}v_T$ and $w_{T^c}=v_{T^c}$. Clearly $w\in\ker(BD)\backslash\{0\}$, but $\|w_{T}\|_1=\|v_{T^c}\|_1+1>\|v_{T^c}\|_1=\|w_{T^c}\|_1$. This concludes that $BD$ does not satisfy \eqref{nsp}.
\end{proof}

\section{Recovery Guarantees for Sparse Signal With Respect to a Frame} \label{sec:recoveryguarantees:frames}
\phantom \\ It can be proven~\cite{LLMLY12, CWW14} that \eqref{equ:F} is equivalent to the so called $\ell_1$-synthesis method:
\begin{equation}\label{equ:l1s}
\Delta_{F, A,\eta}(y):= F(\Delta_{1, AF, \eta}(y)) = F\left(\argmin_x \|x\|_1, \quad\text{subject to }\|AFx - y\|_2\leq\eta\right).
\end{equation}
The idea is that given $y \approx Az_0$, we first use $\Delta_{1, AF, \eta}(y)$ to recover a (hopefully sparse) coefficient of $z_0$. We then recover the original signal by applying the synthesis operator $F$ to the coefficient.

Recalling \eqref{equ:tail},  let
\begin{equation}
z_k:= \argmin_{v\in\Sigma_{F,k}}\|z-v\|_F,
\end{equation}
and therefore $\sigma_{F,k}(z) = \|z - z_k\|_F$.

\begin{definition}\label{def:Fp}
Fix a frame $F\in\R^{d\times n}$. Let $A\in\R^{m\times d}$.
\begin{enumerate}
\item $A$ is said to have the \emph{$F$ null space property} of order $k$~\cite{CWW14, C24} if
\begin{equation}\tag{$F$-NSP-$k$}\label{fnsp}
\|v\|_F<\|z-v\|_F, \ \forall z\in\ker(A)\backslash\{0\}, \ \forall v\in\Sigma_{F,k}
\end{equation}

\item $A$ is said to have the \emph{robust $F$ null space property} of order $k$~\cite{C24} if there exist  constants $\tau>0$ and $0<\rho<1$ such that 
\begin{equation}\label{frnsp}\tag{\ensuremath{F}-RNSP-\ensuremath{(k,\rho,\tau)}}
\|z_k\|_F\leq\rho\|z-z_k\|_F+\tau \|A z\|_2, \ \forall z\in\R^d.
\end{equation}

\item $A$ is said to have the \emph{strong $F$ null space property} of order $k$ with constant $c>0$~\cite{CWW14, C24} if  
\begin{equation}\label{sfnsp}\tag{$F$-SNSP-($k, c$)}
\|z-v\|_F-\|v\|_F\geq c\|z\|_2,  \  \forall z\in\ker(A), \forall v\in\Sigma_{F,k}.
\end{equation}

\item $A$ is said to have the \emph{$F$ restricted isometry property} of order $k$ with constant $0 \leq \delta < 1$~\cite{CENR11} if
\begin{equation}\label{frip}\tag{\ensuremath{F}-RIP-\ensuremath{(k,\delta)}}
(1 - \delta) \|z\|_2^2 \leq \|Az \|^2_2 \leq (1 + \delta) \|z\|_2^2, \ \forall x\in\Sigma_{F,k}.
\end{equation}
\end{enumerate}
\end{definition}

\begin{remark}\label{rem:fs}
An equivalent, but less intuitive  version of \eqref{fnsp} can be found in \cite[Definition 4.1]{CWW14}. Similarly, an equivalent definition of \eqref{sfnsp} can be found in \cite[Definition 5.1]{CWW14}.
\end{remark}

\begin{remark}\label{rem:sfnsp}
By definition, \ref{sfnsp} implies \ref{fnsp}. A remarkable result~\cite[Theorem 5.5]{CWW14} states that \ref{fnsp} implies \ref{sfnsp} for some $c>0$. 
\end{remark}

\begin{remark}\label{rem:equiv}
\cite[Theorem 7.2]{CWW14} states that the following three conditions are equivalent if $F$ is \emph{full spark} (every $d$ columns of $F$ are linearly independent)~\cite{BCM12}:

 (i) $A$ has \ref{fnsp}.

(ii) $A$ has \ref{sfnsp} for some $c>0$.

(iii) $AF$ has \ref{nsp}.

\end{remark}

\begin{theorem}[\cite{CWW14, C24}]\label{thm:fnsp}
$\Delta_{F,A,0}(Az)=z$ for any $z\in\Sigma_{F,k}$ if and only if $A$ satisfies \ref{fnsp}.
\end{theorem}

Theorem \ref{thm:fnsp} is analogous to Theorem \ref{thm:nsp}. For a result such as Theorem \ref{theorem:RNSP:guarantee}, we need the frame $F$ to be somewhat regular in the following way:
\begin{definition}[{\cite[Definition 2.5]{C24}}]\label{def:split}
We call $F$ $s$-splittable with constant $\beta>0$ if for any $x,y\in\R^d$, 
\begin{equation}\label{equ:split}
\|x+y\|_F\geq\|x_s\|_F-\|y_s\|_F+\beta(\|y-y_s\|_F-\|x-x_s\|_F).
\end{equation}
\end{definition}
Requiring $F$ to be regular in the above way is not surprising as the unit ball of $F$-norm \eqref{equ:zF} is the convex hull of columns in $F$ (together with its negatives). In the case where $F$ is the canonical basis, we have $\beta=1$, which is the optimal value. We conjecture that all full-spark frames are splittable with small enough $\beta$, and it would be interesting work to estimate $\beta$ based on certain frame properties of $F$.

\begin{theorem}[{\cite[Corollary 5.3]{C24}}]\label{thm:frnsp}
Let $F$ be $s$-splittable with constant $\beta$ and $A$ has \ref{frnsp} with $\rho<\beta$,
then for any $z\in\R^d$ and  $\|w\|_2\leq\eta$, we have
\begin{equation*}
\|z - \Delta_{F, A, \eta}(A z+w)\|_F\leq\frac{(1+\beta)(1+\rho)}{\beta(\beta-\rho)}\sigma_{F,s}(z)+\frac{2\tau(1+\beta)}{\beta-\rho}\eta.
\end{equation*}
\end{theorem}
The above theorem is a generalization of \eqref{equ:thm:rnsp:1} in Theorem \ref{theorem:RNSP:guarantee} (with $q=1$). It remains to be future work to establish a performance guarantee like \eqref{equ:thm:rnsp:q}, which should involve a modified version of \ref{frnsp}.

\begin{theorem}[{\cite[Corollary 5.4]{C24}}]
If $A$ has \ref{sfnsp}, then for any $z\in\R^d$ and  $\|w\|_2\leq\eta$, we have
\begin{equation}\label{equ:Fsnsp}
\|z - \Delta_{F, A, \eta}(A z+w)\|_2\leq\frac{2}{\nu_A}\left(\frac{\sqrt{n}}{c}+1\right)\eta+\frac{2}{c}\sigma_{F,k}(z),
\end{equation}
where $\nu_A$ is the smallest nonzero singular value of $A$.

\end{theorem}

By \Cref{rem:sfnsp}, \ref{sfnsp} is a minimal condition, which partially explains the weaker performance guarantee than that of Theorem \ref{thm:frnsp} due to the growth term $\sqrt{n}$ in \eqref{equ:Fsnsp}.

It is yet to be researched how the equality-constrained decoder $\Delta_{F, A, 0}$ behaves under reasonable conditions on $A$. As mentioned previously, a result similar to Theorem \ref{thm:equality} is beneficial when noise-level $\eta$ is unknown. 

\vspace{0.75mm}
\paragraph{Analysis Sparsity Models}
Another popular optimization method in the form of \eqref{equ:f} is the $\ell_1$-analysis method:
\begin{equation}\label{equ:analysis}
\argmin_z \|Gz\|_1, \quad\text{subject to }\|Az - y\|_2\leq\eta, 
\end{equation}
where $G$ is sparsifying transformation such that $Gz$ is sparse. Note that this is not necessarily in the general framework \eqref{eq:SigmaFk}, but is a relevant model with wide applications, such as finite difference and wavelet transform for imaging problems~\cite{KrahmerWard-2014}. In particular, if $F$ is a Parseval frame, then $G=F^T$ is a popular choice~\cite{CENR11, ACP12}, and we denote it by 
\begin{equation}\label{equ:l1a}
\Lambda_{F, A,\eta}(y):=\argmin_z \|F^Tz\|_1, \quad\text{subject to }\|Az - y\|_2\leq\eta.
\end{equation}

\begin{theorem}[{\cite[Theorem 1.4]{CENR11}}]\label{thm:frip}
Let $F$ be Parseval\footnote{$FF^T=I$}. If $A$ satisfies $F$-RIP-$(2k,\delta)$ with $\delta<0.08$, then
\begin{equation}\label{equ:thmfrip}
\|z-\Lambda_{F,A,\eta}(Az+w)\|_2\leq D_1\frac{\sigma_k(F^Tz)}{\sqrt{k}}+D_2\eta
\end{equation}
for all $z\in\R^d$ and $w\in\R^m$ with $\|w\|_2\leq\eta$. Both $D_1, D_2$ only depend on $\delta$.
\end{theorem}
An interesting phenomenon with Theorem \ref{thm:frip} is that there is a mismatch between signal model and sensing matrix requirement. The condition \ref{frip} requires $A$ to act as a near isometry on $\Sigma_{F,s}$, so one hopes that   with appropriately chosen $s$ and $\delta$, \ref{frip} is able to provide a performance guarantee in the form of \eqref{equ:performance}. However, \eqref{equ:thmfrip} only guarantees the stable recovery of analysis sparse signals. It has been a long standing question  whether \ref{frip} is able to guarantee the performance of the $\ell_1$-synthesis method \eqref{equ:l1s}. We answer this question negatively in this book chapter, and it is essentially due to the lack of involvement of some $F$-related-norm in the inequality \ref{frip}.

Let us make things simpler for a moment and let $F=D$ be an $d\times d$ invertible diagonal matrix. In this case, $\Sigma_{D, k}=\Sigma_k$, which means \ref{rip} is equivalent to $D$-RIP-$(k,\delta)$. However, the $\ell_1$-synthesis method is highly sensitive to $D$, even when reduced to a diagonal matrix. To be precise, we have $\Delta_{D,A,\eta}(y)=\Delta_{1, AD, \eta}(y)$, whose success of recovering all signals in $\Sigma_k$ requires $AD$ to satisfy \ref{nsp} by Theorem \ref{thm:nsp}, but this cannot be true for all $D$ by Lemma \ref{lem:nsp}.

The above argument can be generalized to full spark  frame $F$ as shown below.
\begin{theorem} \label{thm:dripbad}
Assume $F$ is of full spark, then $A$ satisfying $F$-RIP-$(2k,\delta)$ does not guarantee the exact recovery of all $F$-sparse signals through the $\ell_1$-synthesis method $\Delta_{F, A, 0}$, regardless of how small $\delta$ is.
\end{theorem}
\begin{proof} Fix a non-singular diagonal matrix $D$.
If $A$ satisfying $F$-RIP-$(2k,\delta)$ does guarantee $\Delta_{F, A, 0}(Az_0)=z_0$ for all $z_0\in\Sigma_{F,k}$, then $A$ satisfying $FD$-RIP-$(2k,\delta)$ also guarantees $\Delta_{FD, A, 0}(Az_0)=z_0$ for all $z_0\in\Sigma_{FD,k} = \Sigma_{F, k}$, which is equivalent to $A$ satisfying $FD$-NSP-$k$ by Theorem \ref{thm:fnsp},  and further equivalent to $AFD$ satisfying NSP-$k$ by Remark \ref{rem:equiv}.

However, by Lemma \ref{lem:nsp}, we can pick $D$ such that $AFD$ does not satisfy NSP-$k$, which is a contradiction.
\end{proof}

\section{Sensing Matrix Design}\label{sec:random}
After having characterized abstract conditions for which $\ell_1$-synthesis decoders are able to robustly identify $F$-$k$-sparse vectors from undersampled measurements in \Cref{sec:recoveryguarantees:frames}, we address now the question which particular sensing matrices $A$ are suitable and satisfy conditions as laid out in \Cref{def:basisp} or \Cref{def:Fp}. Apart from the choice of the sensing matrix (question Q2), we are also shedding light on the minimum number $m$ of rows of $A$ (number of measurements) that is sufficient to enable successful recovery. 
In particular, we seek to find desirable $A$  \textbf{with $m$ being as small as possible, while fixing $n$, the number of frame vectors, $d$, the ambient signal dimension,  and $k$, the sparsity level.} 

\subsection{Sensing Matrices for Identity Matrix $F$} \label{sec:sensing:identity}
A core theme in the theory of compressed sensing problems has been the fact that in order to obtain recovery guarantees that hold uniformly across all signals in the desired sparse signal set such as $\Sigma_{F,k}$ while at the same time keeping the number of measurements $m$ close to the complexity of the signal set, random matrix ensembles are more suitable than deterministic designs for $A$. In this section, we present some key results that apply in the standard compressed sensing setting where the frame matrix is the identity matrix such that $F = I_d$. 

While deterministically designed sensing matrices can achieve empirical performance on par with the one of ideal random designs based on Gaussian sensing matrices \cite{Monajemi-2013Deterministic}, it has been notoriously hard to obtain deterministic sensing matrices with $m$ rows that are able to reconstruct $k$-sparse signals in a setting where $m$ depends almost linearly on $k$. The best known results for deterministically constructed sensing matrices requires $A$ to have $m=O(k^2)$ rows using coherence-based analysis~\cite{BG09} or $m = (k^{2-\epsilon})$ rows for some $\epsilon >0$ using a RIP-based analysis \cite{Bourgain-2011ExplicitRIP}. In the latter case, however, we also require the sensing matrix to be almost square such that $d^{1-\epsilon} \leq m$. We refer to \cite{Bandeira-2013Road,Clum-2022Derandomized,Arjoune-2018Performance}, \cite[Chapter 5]{Vidyasagar-2019CS} for an overview for proof techniques and open questions. 

On the other hand, allowing for random designs of $A$, we can achieve much better results such that already $m=O(k\log(d/k))$ measurements are sufficient for robust and stable recovery, if entries, rows, or columns of $A$ are drawn independently from some suitable distributions. Key results  for different distributional assumptions have been obtained in \cite{D06,CT06,CRT06,Baraniuk-2008Simple,MPT08,Kol11,F14,MendelsonLearning15,ML17,AK22}, the implications of which we report below. The order $\Omega(k \log(d/k))$ is known to be necessary, as shown by \cite{Foucart-2010Gelfand}, \cite[Chapter 10]{FR13}, cf. \Cref{thm:lower:bound} below.

\vspace{0.75mm}
\paragraph{Sub-Gaussian Measurements} The most well-studied class of random matrices suitable for compressed sensing, which also includes matrices with Gaussian i.i.d. entries, are sub-Gaussian sensing matrices \cite{MPT08} with rows or columns that are independent sub-Gaussian random vectors, which can be defined as follows. 
\begin{definition}[\cite{T15,Ver18:High-Dimensional-Probability}] \label{def:subG} A random vector $\phi$ is called a \emph{sub-Gaussian vector} with parameter $\sigma$
if there exists  a constant $\sigma > 0$ such that 
 \[
 \Pr\left(|\langle\phi, z\rangle|\geq t\right)\leq 2\exp(-t^2/(2\sigma^2))
 \]
  for each $t \geq 0$ and for every $\ell_2$-unit norm vector $z$.
\end{definition}
Furthermore, we call a random vector \emph{centered} if $\Eb [ \phi] = 0$, and we call it \emph{isotropic} if $\Eb[\phi \phi^{\top}] = I$ where $I$ is the identity matrix \cite[Definition 3.2.1]{Ver18:High-Dimensional-Probability}.

\begin{theorem}[{\cite[Theorem 5.65]{Ver10},\cite[Theorem 9.2]{FR13}}]\label{thm:subrip}
Let $A$ be an $(m\times d)$  matrix with independent sub-Gaussian rows or sub-Gaussian columns $\phi$ (in the latter case, we further assume that $\|\phi\|_2 = \sqrt{m}$ almost surely) of parameter $\sigma$.
 Then there exist constants $C,c>0$, depending only on $\sigma$ from Definition \ref{def:subG}, such that $\frac{1}{\sqrt{m}}A$ satisfies \ref{rip} with probability at least $1-2 \exp(-c m \delta^2)$ provided that
$$m\geq C\delta^{-2} k\log(e d/k).$$
\end{theorem}
Combined with Theorem \ref{thm:rip}, \Cref{thm:subrip} implies that a sub-Gaussian matrix, with $m$ at least on the order of $k\log(d/k)$, with large probability, allows $\ell_1$-minimization $\Delta_{1,A,\eta}$ to recover nearly sparse signals stably and robustly. Conversely, $m=O(k\log(d/k))$ is also needed for any decoder to recover nearly sparse signals, as seen in \Cref{thm:lower:bound}, which can be shown using results about Gelfand widths of $\ell_1$-balls \cite{Garnaev84,Foucart-2010Gelfand}.

\begin{theorem}[{\cite[Proposition 10.7]{FR13}}] \label{thm:lower:bound}
Let $1<p\leq2$. Suppose that there exist a matrix $A\in\R^{m\times d}$ and a decoder $\Delta$ such that for any $z \in \mathbb{R}^d$,
$$\|z-\Delta(Az)\|_p\leq\frac{C}{k^{1-1/p}}\sigma_k(z),$$
then $m\geq c_1k\log(ed/k)$ for some $c_1, c_2>0$ only depending on $C$.
\end{theorem}

The methodology for proving \Cref{thm:subrip} for independent rows has been developed in \cite{MPT08} using a concentration inequality for sub-Gaussian vectors. For the case of independent sub-Gaussian columns, it is convenient to use extremal singular value estimates \cite[Theorem 5.65]{Ver10} and for the special case when $A$ is Gaussian, results with sharper bounds can be obtained in this manner~\cite{CT06}.

Apart from Gaussians, the class of sub-Gaussian distributions of \Cref{def:subG} includes also random vectors defined via independent Bernoulli entries, bounded entries or random vectors uniformly distributed on a unit sphere.

Given \Cref{theorem:RNSP:guarantee}, the optimal number of measurements can also be achieved by proving that sub-Gaussian matrices satisfy a robust null space property such as RNSP-$(q,k,\rho,\tau)$. In fact, this can be considered as a superior route because the robust NSP is weaker than the RIP for the same sparsity level $k$, and hence possibly admits a broader class of random matrices. 

\vspace{0.75mm}
\paragraph{Heavy-Tailed Measurements}
For designing sensing matrices based on random designs, it is a natural question to ask what the weakest possible assumption on the distribution of the random vector $\phi$ contained in its rows are. Using properties such as \ref{nsp} and \ref{rnsp}, it has been possible to show that heavier-tailed matrices with just sub-exponential, but not sub-Gaussian tail decay exhibit similar guarantees as sub-Gaussian sensing matrices, despite the fact that they cannot satisfy a RIP-$(k,\delta)$ unless their number of rows scales as $m = \Omega( k \log^2(e d/k))$ \cite{Adamczak-2011Restricted}. 
A noticeable example for such is matrices with i.i.d. entries that follow a Weibull distribution with scale parameters between $1 \leq r \leq 2$ (which include the Laplace distribution as a special case) and which is sub-exponential \cite{F14}. It has been shown that for such matrices, the RNSP of order $k$ holds with large probability with $m=O(k\log(ed/k)$~\cite{F14}, making such matrices candidates for optimal performance of the decoder $\Delta_{1,A,\eta}$. For exact recovery, i.e., the \ref{nsp} of $A$ is satisfied, it was shown in \cite[Theorem 7.3]{Kol11} that $m=O(k\log(ed/k)$ measurements of a random ensemble defined from vectors distributed as other sub-exponential distributions (including log-concave distributions) are sufficient.

Optimal compressed sensing recovery guarantees for heavy-tailed distributions were subsequently developed and generalized \cite{ML17,DLR18,AK22} using the small-ball method (due to Shahar Mendelson and Vladimir Koltchinskii \cite{Mendelson-2014Remark,Koltchinskii-2015Bounding,MendelsonLearning15}) based on the observation \cite[Theorem B]{ML17} that while the lower inequality in the property \ref{rip} is easily fulfilled for heavy-tailed distribution satisfying a non-degeneracy assumption stating that not the entire distributional mass is located around the origin, it is sufficient to satisfy the upper RIP-inequality for $1$-sparse vectors. 

Instead of sub-Gaussianity, we state below two weaker assumptions that still enable respective sensing matrices to achieve stable and robust recovery from a minimal amount of measurements via $\Delta_{1,A,\eta}$.
\begin{definition}[{Distributions satisfying small-ball assumption, \cite[Definition 1.4]{ML17}, \cite[Inequality (9)]{DLR18}, \cite{AK22}}] \label{def:small-ball} 
We say the distribution of a random vector $\phi \in \mathbb{R}^d$ \emph{satisfies a small-ball assumption on the set $S \subset \R^d$ with constants $u,c > 0$} if 
\[
\inf_{x\in S}  \Pr(|\langle \phi, x\rangle| \geq u) \geq c.
\]
\end{definition}

\begin{definition}[{Distributions with $K$ well-behaved entrywise moments, \cite[Theorem A(1)]{ML17}, \cite[Inequality (12)]{DLR18}, \cite[Assumption 3]{AK22}}] \label{def:finite-moment} 
We say the distribution of a random vector $\phi = (\phi_1,\ldots,\phi_d) \in \mathbb{R}^d$ has \emph{$K$ well-behaved entrywise moments with constants $\alpha \geq 1/2$ and $\lambda > 0$} if for each $i=1,\ldots,d$ and for all $2\leq p\leq K$,
\begin{equation} \label{eq:tail:decay}
\| \phi_i \|_{L^p} \leq \lambda p^{\alpha}.
\end{equation}
\end{definition}
For a random variable $X$, we define $\|X\|_{L^p}:=\Eb(|X|^p)^{1/p}$.
For sensing matrices $A$ with independent rows that are centered and satisfy \Cref{def:small-ball} and \Cref{def:finite-moment}, we state below the result with the weakest known distributional assumptions that still lead to robust recovery via $\Delta_{1,A,\eta}$ in the optimal measurement regime. We refer to \cite[Section 7.2]{FR13} for results quantifying the distributional tail decay based on moment bounds such as \cref{eq:tail:decay}.

\begin{theorem}[{\cite[Variation on Theorem 4.1]{AK22}}] \label{thm:AK22:41} Let $a_i/\sqrt{m}$, $i\in [m]$ be independent rows of the measurement matrix $A$ drawn from centered distributions satisfying \Cref{def:small-ball} on the set 
\[
S_{k,\rho}^2 :=\left\{v\in \mathbb{R}^d: \exists T\subset [d] \textrm{ with } |T|=k \text{ such that }  \|v_T\|_2 \geq \frac{\rho}{\sqrt{k}}\|v_{T^c}\|_1  \right\} \cap \mathbb{S}_{\ell_2}^{d-1}
\]
with constants $u,c >0$ and \Cref{def:finite-moment} with constants $K = \log( e d / k)$, $c > 0$ and $\alpha \geq 1/2$ and $\lambda > 0$, respectively. Then as long as the number of measurements $m$ satisfies
\begin{equation} \label{eq:measurement:bound:finitemomentresult}
m \gtrsim \max\left\{k\log(ed/k),\log^{\max \{2\alpha-1, 1\}}(ed/k)\right\},
\end{equation}
it holds that with probability of at least $1-e^{-\Omega(m)}$, the decoder $ \Delta_{1, A,\eta}$ provides stable and robust recovery of all $z \in \R^d$ from $y=Az+w$ with $\|w\|_2\leq \eta$ such that
\begin{equation} \label{eq:weakmoment:recoveryguarantee}
\|z- \Delta_{1, A,\eta}(A z + w)\|_p \leq C \frac{\sigma_s(z)}{k^{1-1/p}} + D k^{\frac{1}{p}-\frac{1}{2}} \eta,
\end{equation}
for each $1 \leq p \leq 2$, where the constants $C, D> 0$ only depend on $u,c,\alpha$ and $\lambda$.
\end{theorem}
The recovery guarantee \eqref{eq:weakmoment:recoveryguarantee} matches the RIP-based ones and the best known one of \eqref{equ:thm:rnsp:q}, and its proof can be written by establishing a RNSP-$(2,k,\rho,\tau)$ with high probability, combined with \Cref{theorem:RNSP:guarantee}. We refer to \Cref{sec:proof:sketch} for a proof sketch.

The assumption on a bound of only $\log( e d / k)$ leading moments in \Cref{thm:AK22:41} is weaker than the ones of the similar results \cite[Corollary 8]{DLR18} and \cite[Theorem A]{ML17} (in the latter case, with a little caveat): Corollary 8 of \cite{DLR18} required \emph{independence} of the coordinates $\phi_1,\ldots,\phi_d$ of the random vector $\phi$ constituting the rows of the sensing matrix, and furthermore, requires $K = \log(d)$ moments with the growth condition \eqref{eq:tail:decay} instead of $\log(e d/ k)$, as well as the small-ball condition of \Cref{def:small-ball} on the entire $\ell_2$-sphere $\mathbb{S}_{\ell_2}^{d-1}$ instead of only on $S_{k,\rho}^2$\footnote{On the other hand, a close inspection of their proof technique shows that a restriction to $S_{k,\rho}^2$ would also work.}. Furthermore, \Cref{thm:AK22:41} has the smaller exponent of $2\alpha -1$ instead of $4 \beta -1$ in \cite[Theorem A]{ML17} in the measurement bound \eqref{eq:measurement:bound:finitemomentresult}, and also requires $O(\log(d))$ well-behaved moments; the small-ball condition of \cite[Theorem A]{ML17} is, on the other hand, weaker than the one of \Cref{thm:AK22:41} as \Cref{def:small-ball} is only required on the set $\{z \in \R^d : \|z\|_0 \leq k\} \cap \mathbb{S}_{\ell_2}^{d-1}$, which is a subset of $S_{k,\rho}^2$.

We note that a small-ball assumption as in \Cref{def:small-ball} is rarely restrictive: It essentially means that not too much probability mass is assigned to $0$ in the distributional marginals defined by the set $S$, and can be derived from weak moment bounds in marginal direction such as $\|\langle \phi, t \rangle \|_{L^{2+\epsilon}} \leq \kappa \|\langle \phi, t \rangle \|_{L^{2}}$ or $\|\langle \phi, t \rangle \|_{L^{2}} \leq \kappa \|\langle \phi, t \rangle \|_{L^{1}}$ for all $t \in S$, or directly from \Cref{def:finite-moment} if the entries of $\phi$ are independent (cf. proof of \cite[Theorem 3.2]{AK22}, in the case of  $\phi$ is isotropic, via the Paley-Zygmund inequality \cite[Lemma 7.16]{FR13}.

\Cref{def:finite-moment} for $\alpha = 1/2$ can be interpreted as requiring the first $K$ moments to grow like the moments of sub-Gaussian random variables, however, moments beyond the $K$-th moment do not even need to exist. Examples for design matrices covered by \Cref{thm:AK22:41} include Student-$t$ random variables of degree $K=\lceil \log( e d / k)\rceil$, which do not have finite moments beyond  degree $K$.

\vspace{0.75mm}
\paragraph{Structured Random Matrices}
In engineering applications of compressed sensing, the usage of fully randomized measurement matrices--Gaussian, sub-Gaussian or heavy-tailed ones--is in most prohibitive due to physical and hardware constraints \cite{Krahmer-2014structured,Adcock-2021Compressive}. Depending on the application, a specifically structured sensing process is required. For example, $A$ could correspond to partial Fourier measurements for magnetic resonance imaging~\cite{LDP07}, to a partial circular convolution for radar imaging~\cite{RRT12, H10} and deblurring~\cite{MRW18}, or binary measurements \cite{Thesing-2021Non}. 

The work by Rudelson and Vershynin~\cite{RV08} states that a random selection of an order of  $O(k\log^3(k)\log(d))$ Fourier measurements are enough to recover signals in $\Sigma_k$. This number is improved in \cite{CDTW18}, but not likely to reach the optimal number $k\log(d/k)$. Another remarkable result by Krahmer and Ward~\cite{KrahmerWard-2014} uses the analysis sparsity model \eqref{equ:analysis} with $G$ being the finite difference operator. They show that $O(k\log^3(k)\log^5(d))$ random partial Fourier measurements are sufficient.

For subsampled convolution, the sensing matrix $A$ can be expressed as
$Az = (z*c)_{\Omega}$, where $c$ is the generator and $\Omega$ indicates the index for subsampling. It is shown in~\cite{MRW18}  that if $c$ is  sub-Gaussian,  then $O(k\log(d/k))$ subsampled random convolution measurements (if $k\lesssim \sqrt{d/\log(d)}$) are enough to satisfy the RNSP, hence ensure the success of $\Delta_{1,A,\eta}$. However, it is worth noting that the signals must be sparse in the canonical basis and it remains an open problem to obtain the optimal number of measurements for signals sparse in any orthonormal basis. The proof technique utilizes the small-ball property and moment estimates of sub-Gaussian vectors. It would  be very interesting to generalize their results for heavier-tailed generators.

\subsection{Sensing Matrices for General Frames $F$}
Compared to the setup studied in \Cref{sec:sensing:identity}, sensing matrix design in the presence of general rectangular frame matrices $F \in \R^{d \times n}$ has been less explored. In this section, we will restrict ourselves to the case that $F$ is full spark, motivated by Remark \ref{rem:equiv}.
Full spark matrices include many common dictionaries (Wavelet, Curvelet, Discrete Cosine Transform, etc.) and will be obtained with probability $1$ if $F$ is randomly drawn from some continuous distribution family. 

Given a full spark dictionary $F$, the goal is to find sensing matrices $A$ that have as few rows as possible while still ensuring exact recovery of signals in $\Sigma_{F,k}$ via the $\ell_1$-synthesis decoder $\Delta_{F, A,\eta}$. 
Due to Remark \ref{rem:equiv}, a premise for such $A$ to exist is $F$ satisfying  NSP-$k$. In other words, if a full spark $F$ does not have NSP-$k$, then the exact recovery of signals in $\Sigma_{F,k}$ is impossible with $\ell_1$-synthesis. This is because Remark \ref{rem:equiv} states that when $F$ is full spark, $AF$ must satisfy the NSP-$k$ for uniform recovery of signals in $\Sigma_{F,k}$ due to \Cref{thm:fnsp}, which in turn means $F$ itself must satisfy the NSP-$k$.

To be able to obtain robust and stable recovery guarantees, we work with a slightly stronger version of the $\ell_2$-robust null space property RNSP-$(2,k,\rho,\tau)$ of \Cref{def:basisp} defined in \Cref{def:robust}, which is convenient to work with for deriving robust recovery results of the $\ell_1$-synthesis decoder $\Delta_{F, A,\eta}$.

\begin{definition}(Robust-NSP* \cite{DLR18,AK22})\label{def:robust} Given $1 \leq q \leq 2$, $F$ is said to fulfill the $\ell_q$-robust null space property* of order $k$ with constant $0<\rho<1$ and $\tau>0$, also called RNSP*-$(q,k,\rho,\tau)$, if 
\begin{equation} \label{eq:robustNSPstar:lower}
\inf_{v\in S_{k,\rho}^q} \|Fv\|_2 \geq \tau^{-1}
\end{equation}
where 
\[
S_{k,\rho}^q :=\left\{v\in \mathbb{R}^n: \exists T\subset [n] \textrm{ with } |T|=k \text{ such that }  \|v_T\|_q \geq \frac{\rho}{k^{1-1/q}}\|v_{T^c}\|_1  \right\} \cap \mathbb{S}_{\ell_q}^{n-1}.
\]
\end{definition} 

By a compactness argument, it follows that there exists a $\tau >0$ such that $F$ satisfies RNSP*-$(2,k,\rho,\tau)$ provided that $F$ satisfies a \emph{stable} null space property, i.e., there exists a $0 < \rho < 1$ such that $\|v_T\|_2 \leq \frac{\rho}{\sqrt k}\|v_{T^c}\|_1$ for all $v \in \textrm{Ker}(F)\backslash\{0\}$ and all $k$-sparse support sets $T \subset [n]$. This condition can in turn be used to derive optimal recovery guarantees for equality-constrained basis pursuit (by considering $F$ as the sensing matrix) from exact measurements and is applicable to both sparse and approximately sparse vectors \cite[Chapter 4]{FR13}. The RNSP*-$(2,k,\rho,\tau)$ with a fixed constant $\tau$ for random frames $F$ can be established with high probability if the rows of $F$ are i.i.d. and drawn from well-behaved distributions and $d \gtrsim  k\log (n/k)$ \cite[Corollary 8]{DLR18}, \cite[Theorem 4.1]{AK22}. 

In fact, it is easy to verify that the RNSP*-$(q,k,\rho,\tau)$ implies the RNSP-$(q,k,\rho,\tau)$ for each set of parameters.
\begin{proposition} \label{prop:robustnsp:implication}
If a matrix $F \in \R^{d \times n}$ fulfills RNSP*-$(q,k,\rho,\tau)$, i.e., the $\ell_q$-robust null space property* of order $k$ with constant $0<\rho<1$ and $\tau>0$ from \Cref{def:robust}, then it also fulfills the RNSP-$(q,k,\rho,\tau)$ of \Cref{def:basisp}.
\end{proposition}
We recall the simple argument from the proof of \cite[Theorem 3]{DLR18}.
\begin{proof}[{Proof of \Cref{prop:robustnsp:implication}}]
	Assume that $F \in \R^{d \times n}$ fulfills RNSP*-$(q,k,\rho,\tau)$. Let us first assume $v \in \R^n$ be such that $\|F v\|_2 < \|v\|_q / \tau$. In this case, we see that 
	\[
	\Big\|F \frac{v}{\|v\|_q}\Big\|_2 < \frac{1}{\tau},
	\]
	which means that $v / \|v\|_q \in \mathbb{S}_{\ell_q}^{n-1}$ cannot be contained in the set $S_{k,\rho}$, which implies that for any subset $T \subset [n]$ with $|T| = k$, it holds that
	\[
	\|v_T\|_q < \frac{\rho}{k^{1-1/q}}\|v_{T^c}\|_1 \leq \frac{\rho}{k^{1-1/q}}\|v_{T^c}\|_1 + \tau \|F v\|_2.
	\]
	On the other hand, if $v \in \R^n$ is such that $\|F v\|_2 \geq \|v\|_q / \tau$, then it holds that
	\[
	\|v_T\|_q \leq \|v\|_q \leq \tau \|F v\|_2 \leq  \frac{\rho}{k^{1-1/q}}\|v_{T^c}\|_1 +\tau \|F v\|_2.
	\]
	Thus, the defining inequality of RNSP-$(q,k,\rho,\tau)$ holds for each $v \in \R^n$.
\end{proof}

For any frame $F$ satisfying the RNSP-$(2,k,\rho,\tau)$, it can be shown that robust and stable recovery guarantees of the $\ell_1$-synthesis decoder $\Delta_{F, A, \eta}$ hold for a minimal number of rows of $A$ if $A$ is a sub-Gaussian sensing matrix \cite{CCL20}.

\begin{theorem}[{see \cite[Corollary 3.6]{CCL20}, \cite{AK22}}] \label{thm:Cor36:CCL20}
Suppose that rows of $A$ are i.i.d., sub-Gaussian vectors with parameters $\sigma$ and that $F$ is of  full spark satisfying RNSP*-$(2,k,\rho,\tau)$ and its columns satisfy $\max\{\|f_i\|^2_2: i\in [n]\} \leq \theta = O(1)$. Then as long as the number of measurements $m$ satisfies
\[
m \gtrsim  k\log (n/k),
\]
the $\ell_1$-synthesis decoder $ \Delta_{F, A,\eta}$ provides stable and robust recovery of both the coefficient vector $x \in \R^n$ and the signal $z \in \R^d$ from $y=Az+w= AFx+w$ with $\|w\|_2\leq \eta$ such that
\begin{equation} \label{eq:coefficient:synthesis:subGaussian}
\|x - \Delta_{1, A F,\eta}(A F x + w)\|_2 \lesssim \frac{\sigma_k(x)}{\sqrt{k}}+ \frac{\tau}{\sigma}\eta
\end{equation}
and
\begin{equation} \label{eq:signal:synthesis:subGaussian}
\|z - \Delta_{F, A,\eta}(A z + w)\|_2 \lesssim \|F\|_2 \left(\frac{\sigma_k(x)}{\sqrt{k}} + \frac{\tau}{\sigma}\eta \right)
\end{equation}
with probability at least $1-e^{-\Omega(m)}$.
\end{theorem}
\Cref{thm:Cor36:CCL20} differs from \cite[Corollary 3.6]{CCL20} in that the latter uses a RNSP*-$(1,k,\rho,\tau)$ assumption on $F$ (under the name \emph{stable NSP}) instead of RNSP*-$(2,k,\rho,\tau)$. Following \cite{AK22}, we use here the RNSP*-$(2,k,\rho,\tau)$, together with \Cref{theorem:RNSP:guarantee} and \Cref{prop:robustnsp:implication}, to obtain the optimal order dependence of the summand measuring the sparsity order misfit through $\sigma_k(x)$, for which the RNSP*-$(1,k,\rho,\tau)$ is not sufficient.
 
Akin to the case of a signal model $\Sigma_{I_d,k}$ of $k$-sparse vectors in the standard basis, it has turned out that heavy-tailed sensing matrices $A$ beyond the sub-Gaussian can also lead to similar guarantees.  In \cite{AK22}, the conclusions of \Cref{thm:Cor36:CCL20} have been generalized to include sensing matrices with independent rows whose distribution is only required to have $O(\log(e n / s))$ well-behaved moments. In particular, we use the following definition, which is a generalization of \Cref{def:finite-moment}.

\begin{definition}[{Distributions with $K$ well-behaved spherical moments}] \label{def:finite-moment:spherical} We say the distribution of a random vector $\phi \in \mathbb{R}^d$ has \emph{$K$ well-behaved spherical moments} with constants $\alpha \geq 1/2$ and $\lambda > 0$ in direction of a $\ell_2$-unit norm vector $a\in \mathbb{S}^{d-1}$ if for all $2\leq p\leq K$,
\begin{equation} \label{eq:finite-moment:bound:spherical}
\|\langle \phi,a \rangle\|_{L^p} \leq \lambda p^{\alpha}.
\end{equation}
\end{definition} 
For measurement matrices whose rows are drawn from distributions satisfying Definition \ref{def:finite-moment:spherical}, the resulting recovery guarantee applicable at the minimal number of rows is stated in the following theorem.

\begin{theorem}[{\cite[Variant of Theorem 3.1]{AK22}}] \label{thm:frame:heavytailed}
 Let $a_i/\sqrt{m}$, $i\in [m]$ be the rows of the measurement matrix $A$ and the full spark frame $F \in \R^{d \times n}$ be given. Suppose the $a_i$ are independent realizations from some centered distributions satisfying \Cref{def:small-ball} for $\mathbb{S}_{\ell_2}^{n-1}$ with constants $u, c > 0$, and that \Cref{def:finite-moment:spherical} holds in the direction of all frame columns $f_j/\|f_j\|_2$, $j=1,\ldots n$, up to $K=\Omega\left(\log\left(\frac{en}{k}\right) \right)$ moments with moment growth parameter  $\alpha\geq 1/2$ and constant $\lambda >0$. Suppose $F$ satisfies the RNSP*-$(2,k,\rho,\tau)$ and its columns satisfy $\max\{\|f_i\|^2_i: i\in [n]\} \leq \theta$. Then as long as the number of measurements $m$ satisfies
\[
m \gtrsim \max\left\{k\log(en/k),\log^{\max \{2\alpha-1, 1\}}(en/k)\right\},
\]
we have with probability at least $1-e^{-\Omega(m)}$, the $\ell_1$-synthesis  decoder $ \Delta_{F, A,\eta}$ provides stable and robust recovery of both the coefficient vector and the signal from $y=Az+w= AFx+w$ with $\|w\|_2\leq \eta$ such that
\[
\|x- \Delta_{1, A F,\eta}(A F x + w)\|_2 \lesssim \frac{\sigma_s(x)}{\sqrt k} + \tau\eta
\]
and
\[
\|z - \Delta_{F, A,\eta}(A z + w)\|_2 \lesssim \|F\|_2 \left(\frac{\sigma_s(x)}{\sqrt k} + \tau\eta \right).
\] 
\end{theorem}

The first three properties in \Cref{def:Fp} are set up to recover the signal only, without necessarily recovering the sparse coefficient. Supported by numerical experiments and a geometric non-uniform phase transition analysis, it has been argued in \cite{MBKW23} that a recovery of $z \in \Sigma_{F,k}$ is possible robustly in setups where the sparse coefficient vector cannot be identified, which in some sense extends the scope of applicability of $\ell_1$-synthesis techniques.

The results presented in this section, \Cref{thm:Cor36:CCL20} and \Cref{thm:frame:heavytailed}, although achieving the optimal number of measurements in a sense, both require the null space property on $F$, and thus recover the signal and coefficient simultaneously. This is essentially due to the challenge of analyzing the null space type properties in Definition \ref{def:Fp} directly, and the existing work has resorted to applying properties in Definition \ref{def:basisp} to $AF$. To show certain class of random sensing matrices satisfy \ref{frnsp} or a variation of it is still an open problem.

Following a different viewpoint, supported by numerical experiments and a geometric non-uniform phase transition analysis, it has been argued in \cite{MBKW23} that a robust recovery of $z \in \Sigma_{F,k}$ is possible  in setups where the sparse coefficient vector cannot be identified, which in some sense extends the scope of applicability of $\ell_1$-synthesis techniques. A current limitation of the analysis in \cite{MBKW23} is that compared to the above theorems it only applies to sub-Gaussian sensing matrices $A$. It remains an open problem to extend such a signal-centric non-uniform analysis to a broader class of random and structured sensing matrices.

\subsection{A Proof Sketch For $\ell_1$-Synthesis Recovery Guarantees} \label{sec:proof:sketch}
We now briefly sketch some of the key steps that can be taken to prove recovery guarantees for $\ell_1$-synthesis decoders of the type $\Delta_{F, A,\eta}$ as presented in \Cref{thm:Cor36:CCL20}, \Cref{thm:frame:heavytailed} and \Cref{thm:AK22:41}.

The starting point is based on the idea that once we consider \Cref{def:robust} of the RNSP*-$(q,k,\rho,\tau)$ of the measurement-frame matrix product $A F$, we can use a technique for lower bounding non-negative empirical processes developed in the series of papers \cite{Mendelson-2014Remark,Koltchinskii-2015Bounding,MendelsonLearning15} to obtain non-trivial lower bound of $\inf_{v \in S_{k,\rho}^2} \|AFv\|_2$. This gives rise to \eqref{eq:robustNSPstar:lower} with a bounding constant that holds with high probability over the draws of $A$, which would correspond to an RNSP*-$(2,k,\rho,\tau)$ for $A F$ and enables robust recovery guarantees for $\ell_1$-synthesis through \Cref{prop:robustnsp:implication}.
\begin{proposition}[{\cite[Theorem 1.5]{Koltchinskii-2015Bounding}, \cite[Proposition 5.1]{T15}}]
\label{small_ball_method}
Fix a set $S \subset \mathbb{R}^d$. Let $\phi \in \mathbb{R}^d$ be a random vector and let $\Phi \in \mathbb{R}^{m\times d}$ be a random matrix whose rows are i.i.d copies of $\phi$. Then, for any $t>0$ and $u>0$,
\begin{equation} \label{eq:prop:sb}
    \inf_{x\in S} \|\Phi x\|_2 \ge u \sqrt{m} Q_{2u}(S;\phi) -2 W_m(S;\phi) - u t,
\end{equation}
with probability at least $1-e^{-t^2/2}$, where $
W_m(S,\phi) := \mathbb{E}\sup_{x\in S} \Big\langle x, \frac{1}{\sqrt{m}} \sum_{i=1}^m \varepsilon_i \phi_i \Big\rangle
$
is the \emph{mean empirical width} of $S$ with respect to $\phi$ if $\varepsilon_1,\ldots,\varepsilon_m$ are independent Rademacher random variables (attaining values $\pm 1$ with probability $1/2$), and $Q_{u}(S;\phi):= \inf_{x \in S} \Pr(|\langle \phi ,x \rangle| \ge u)$ a small-ball probability bound of $\phi$ with respect to the set $S$ (cf. also \Cref{def:small-ball}).
\end{proposition}
\Cref{small_ball_method} can be shown by estimating, for fixed $x\in S$, that $\|\Phi x\|_2 \geq  \frac{1}{\sqrt{m}} \|\Phi x\|_1 =\frac{1}{\sqrt{m}} \sum_{i=1}^m |\langle \phi_i,x\rangle | \geq \frac{u}{\sqrt{m}}\sum_{i=1}^m \mathds{1}_{\{|\langle \phi_i,x\rangle | \geq u\}}$ where $\phi_1,\ldots,\phi_m$ are independent copies of $\phi$ and $\mathds{1}_{B}$ a $1$-$0$ random variable indicating whether the $B$ takes places or not, comparing the resulting bound with $u \sqrt{m} Q_{2u}(S;\phi) $ and using the bounded differences inequality \cite[Section 6.1]{Boucheron-book2013}; a transparent proof can be found in \cite[Section 5.5]{T15}.

With \Cref{small_ball_method} applied to $\Phi= A$ and $S= F S_{k,\rho}^2$, it remains to lower bound $Q_{2u}(F S_{k,\rho}^2;\phi)$ for appropriately chosen constant $u$ and upper bound the empirical mean width $W_m(F S_{k,\rho}^2,\phi)$.

For the former, we can calculate that
\begin{equation*}
\begin{split}
&Q_{2u}(F S_{k,\rho}^2;\phi) = \inf_{x \in F S_{k,\rho}^2} \Pr(|\langle \phi ,x \rangle| \ge u) = \inf_{v\in S_{k,\rho}^2} \Pr(|\langle \phi ,F v \rangle| \ge u) \\
&= \inf_{v\in S_{k,\rho}^2} \Pr\Big(\Big|\Big\langle \phi ,\frac{F v}{\|F v\|_2} \Big\rangle\Big| \ge \frac{u}{\|F v\|_2}  \Big) \geq \inf_{v\in S_{k,\rho}^2} \Pr\Big(\Big|\Big\langle \phi ,\frac{F v}{\|F v\|_2} \Big\rangle\Big| \ge \tau u  \Big)
\end{split}
\end{equation*}
using the assumption that $F$ satisfies RNSP*-$(2,k,\rho,\tau)$ (the last inequality is not needed for the case $F = I_d$). The latter term can be further lower bounded by the constant $c$ from \Cref{def:small-ball} since  $\frac{F v}{\|F v\|_2} \in \mathbb{S}_{\ell_2}^{n-1}$ in the case of \Cref{thm:frame:heavytailed}, and lower bounded by an absolute constant through the Paley-Zygmund inequality \cite[Lemma 7.16]{FR13} in the sub-Gaussian case of \Cref{thm:Cor36:CCL20}.

For the empirical mean width of the set $F S_{k,\rho}^2$, we can use the result \cite[Lemma 2]{DLR18} that implies the geometric inclusion $S_{k,\rho}^2 \subset (2 + \rho^{-1}) \conv\big(Z_k\big)$ with $Z_k :=\{z \in \R^n : \|z\|_0 \leq k, \|z\|_2 =1\}$, where $\conv(\cdot)$ denote the convex hull of the respective set. With the notation $V:= m^{-1/2}\sum_{i=1}^m \varepsilon_i\phi_i$, this leads to estimate that
\begin{equation*}
\begin{split}
	&W_m(F S_{k,\rho}^2,\phi) = \mathbb{E}\sup_{x\in F S_{k,\rho}^2} \big\langle x, V \big\rangle =  \mathbb{E}\sup_{x\in S_{k,\rho}^2} \big\langle x, F^{\top}V \big\rangle \\
	&\leq (2 + \rho^{-1}) \mathbb{E}\sup_{x\in \conv(Z_k)} \big\langle x, F^{\top}V \big\rangle \\
	&=  (2 + \rho^{-1}) \mathbb{E}\sup_{x \in Z_k} \big\langle x, F^{\top}V \big\rangle = (2 + \rho^{-1})  \left(\sum_{i=1}^k  ((F^{\top}V)_i^*)^2 \right)^{1/2},
\end{split}
\end{equation*}
using the fact that the supremum of the linear form over a set and over its convex hull coincide; for a vector $z=(z_i,\ldots,z_n)$, we denoted the $i$-th largest coordinate in absolute value as $z_i^*$. In the case of sub-Gaussian vectors, it is possible to invoke the majorizing measure theorem \cite[Theorem 2.7.2 and Theorem 2.10.1]{Talagrand-2021Upper} to bound $\mathbb{E}\sup_{x \in Z_k} \big\langle x, F^{\top}V \big\rangle$ by a constant times the Gaussian width of the set $F Z_k$, as used in \cite{CCL20}, which can be bounded explicitly by properties of the Gaussian distribution.

For heavy-tailed distributions, however, one can use the weak moment assumption on inner products of the $\phi_i$ with the columns $f_j$ for the frame matrix $F$ (\Cref{def:finite-moment:spherical}) in conjunction with the following elementary lemma applied to $z=(\langle f_1 ,V\rangle,\ldots,\langle f_n ,V\rangle)$, which was shown in \cite{AK22} and which relaxes the assumptions of \cite[Lemma 6.5]{MendelsonLearning15}.
\begin{lemma}[{Bound on Order Statistics Norm, \cite[Lemma 5.1]{AK22}}] 
There exists an absolute constant $C > 0$ for which the following holds. Let $k \in \N$. Assume that $z_1,\ldots,z_n$ are centered random variables with  variance $1$ that fulfill for each $i \in [n]$ that for every $p \leq 2 \log(n/k)$, $\|z_i\|_{L_p} \leq \lambda \sqrt{p}$.
Then
\begin{equation*}
\mathbb{E}\left[\sum_{i=1}^k \left(z_i^{*}\right)^2\right]^{1/2}\leq C \lambda \sqrt{ k\log\left(\frac{n}{k}\right)},
\end{equation*}
where $z_i^{*}$ denotes the $i$-th coordinate of the non-increasing rearrangement of the vector $z = (|z_1|,\ldots,|z_n|)$.
\end{lemma}
Putting these proof ingredients together is suitable to establish \Cref{thm:AK22:41,thm:Cor36:CCL20,thm:frame:heavytailed}.

\vspace{0.75mm}
\paragraph{Open Questions} To the best of our knowledge, results on recovery guarantees for dictionary or frame sparse signals for structured design matrices for minimum measurement numbers are rare and hardly existent, especially for $\ell_1$-synthesis decoders. It would be interesting to investigate if some of the presented techniques can be adapted to analyze designs based on partial bounded orthogonal systems or partial circulant matrices in the frame-sparse setup. A generalization of the analysis of subsampled convolutions of~\cite{MRW18} to this setup, potentially even with heavy-tailed generator vector, might be in reach due to the similarity of their techniques, but remains open at this point.

While not a focus of this chapter, we note that also for $\ell_1$-analysis decoders such as $\Lambda_{F, A,\eta}$ from \eqref{equ:l1a}, there are only few strong results about structured random measurements available, which are based on the \eqref{frip} of \Cref{def:Fp}, such as those of \cite{Krahmer-2015Compressive}. It would be interesting to explore whether NSP-based proof strategies such as presented in this chapter leads to guarantees applicable for fewer measurements than $F$-RIP based ones for $\ell_1$-analysis decoders.

\section{Outlook}
In this chapter, we presented results about the interplay of decoder, sensing matrix properties enabling robust recovery and discussed the question of sensing matrix design in this context, for the problem of identifying signals that are linear combinations of few dictionary atoms or frame vectors. We also presented several open research directions. By focusing on convex decoders, our survey cannot do justice to the analysis of practically very relevant non-convex and iterative decoders; however, the comparatively mature theory of convex decoders can be seen as a starting point for further developing the analysis of such decoders as well.

Related to the sparse recovery problems we studied are also low-rank matrix recovery problems \cite{Davenport16,CaiWeiExploiting18,chen_chi18,ChiLuChen19,Fuchs-2022Proof}, non-linear variants of which have gained renewed interest in the machine learning community \cite{Hu-2022LoRA,Zeng-2024ExpressiveLoRA,Jang-2024LoRA}. While robust recovery guarantees based on (robust) null space properties are also available for such problems \cite{Yi-2020Necessary}, existing ones are not applicable in a uniform manner for many important special cases such as entrywise sampling (matrix completion). However, we note that relevant problem variants can be considered as \emph{Schatten-$1$-synthesis} or \emph{Schatten-$1$-analysis} problems in analogy to the concepts presented here. Examples are the Euclidean distance geometry problems, whose low-rank optimization perspective can be considered as a nuclear (or Schatten-$1$) norm minimization problem with respect to a synthesis operator that maps Gram matrices to pairwise distance matrices \cite{Dokmanic-2015,Tasissa-2018EDG}, and low-rank (block-)Hankel/Toeplitz recovery problems \cite{ChenChi-2014,Jin-2016General,Cai-2019Fast}, which can be framed as low-rank recovery problems with respect to a Hankel or Toeplitz-type analysis operator. Exploring the benefits of this perspective remains for future investigations.

\section*{Acknowledgements}
XC is partially funded by NSF grant 2307827.

\bibliographystyle{amsplain}
\bibliography{ref_24_08} 

\providecommand{\bysame}{\leavevmode\hbox to3em{\hrulefill}\thinspace}
\providecommand{\MR}{\relax\ifhmode\unskip\space\fi MR }
\providecommand{\MRhref}[2]{%
  \href{http://www.ams.org/mathscinet-getitem?mr=#1}{#2}
}
\providecommand{\href}[2]{#2}
\begin{thebibliography}{100}

\bibitem{AK22}
Pedro Abdalla and Christian K{\"u}mmerle, \emph{Dictionary-sparse recovery from
  heavy-tailed measurements}, Information and Inference: A Journal of the IMA
  \textbf{11} (2022), no.~4, 1501--1526.

\bibitem{Adamczak-2011Restricted}
Radoslaw Adamczak, Alexander~E. Litvak, Alain Pajor, and Nicole
  Tomczak-Jaegermann, \emph{Restricted isometry property of matrices with
  independent columns and neighborly polytopes by random sampling},
  Constructive Approximation \textbf{34} (2011), 61--88.

\bibitem{Adcock-2021Compressive}
Ben Adcock and Anders~C. Hansen, \emph{Compressive imaging: Structure,
  sampling, learning}, Cambridge University Press, 2021.

\bibitem{ACP11}
Akram Aldroubi, Xuemei Chen, and Alex Powell, \emph{Stability and robustness of
  $\ell_q$ minimization using null space property}, Proceedings of the 9th
  International Conference on Sampling Theory and Applications (SampTA),
  Singapore, 2011.

\bibitem{ACP12}
Akram Aldroubi, Xuemei Chen, and Alexander~M. Powell, \emph{Perturbations of
  measurement matrices and dictionaries in compressed sensing}, Appl. Comput.
  Harmon. Anal. \textbf{33} (2012), no.~2, 282--291.

\bibitem{BCM12}
Boris Alexeev, Jameson Cahill, and Dustin~G. Mixon, \emph{Full {S}park
  {F}rames}, J. Fourier Anal. Appl. \textbf{18} (2012), no.~6, 1167--1194.
  \MR{3000979}

\bibitem{Arjoune-2018Performance}
Youness Arjoune, Naima Kaabouch, Hassan El~Ghazi, and Ahmed Tamtaoui, \emph{A
  performance comparison of measurement matrices in compressive sensing},
  International Journal of Communication Systems \textbf{31} (2018), no.~10,
  e3576.

\bibitem{Bandeira-2013Road}
Afonso~S. Bandeira, Matthew Fickus, Dustin~G. Mixon, and Percy Wong, \emph{The
  road to deterministic matrices with the restricted isometry property},
  Journal of Fourier Analysis and Applications \textbf{19} (2013), no.~6,
  1123--1149.

\bibitem{Baraniuk-2008Simple}
Richard Baraniuk, Mark Davenport, Ronald DeVore, and Michael Wakin, \emph{A
  simple proof of the restricted isometry property for random matrices},
  Constructive approximation \textbf{28} (2008), 253--263.

\bibitem{Baraniuk-2017Compressive}
Richard Baraniuk and Philippe Steeghs, \emph{Compressive sensing: A new
  approach to seismic data acquisition}, The Leading Edge \textbf{36} (2017),
  no.~8, 642--645.

\bibitem{Boucheron-book2013}
St{\'e}phane Boucheron, G{\'a}bor Lugosi, and Pascal Massart,
  \emph{Concentration inequalities: A nonasymptotic theory of independence},
  Oxford university press, 2013.

\bibitem{Bourgain-2011ExplicitRIP}
Jean Bourgain, Stephen Dilworth, Kevin Ford, Sergei Konyagin, and Denka
  Kutzarova, \emph{{Explicit constructions of RIP matrices and related
  problems}}, Duke Mathematical Journal \textbf{159} (2011), no.~1, 145 -- 185.

\bibitem{BG09}
Jean Bourgain and M.~Z. Garaev, \emph{On a variant of sum-product estimates and
  explicit exponential sum bounds in prime fields}, Mathematical Proceedings of
  the Cambridge Philosophical Society, vol. 146, Cambridge University Press,
  2009, pp.~1--21.

\bibitem{Brugiapaglia-2018Robustness}
Simone Brugiapaglia and Ben Adcock, \emph{Robustness to unknown error in sparse
  regularization}, IEEE Transactions on Information Theory \textbf{64} (2018),
  no.~10, 6638--6661.

\bibitem{CCW16}
Jameson Cahill, Xuemei Chen, and Rongrong Wang, \emph{The gap between the null
  space property and the restricted isometry property}, Linear Algebra and its
  Applications \textbf{501} (2016), 363--375.

\bibitem{CM21}
Jameson Cahill and Dustin~G. Mixon, \emph{Robust width: A characterization of
  uniformly stable and robust compressed sensing}, Excursions in Harmonic
  Analysis, Volume 6: In Honor of John Benedetto's 80th Birthday (2021),
  343--371.

\bibitem{CCS08}
Jian-Feng Cai, Raymond~H Chan, and Zuowei Shen, \emph{A framelet-based image
  inpainting algorithm}, Applied and Computational Harmonic Analysis
  \textbf{24} (2008), no.~2, 131--149.

\bibitem{Cai-2019Fast}
Jian-Feng Cai, Tianming Wang, and Ke~Wei, \emph{{Fast and provable algorithms
  for spectrally sparse signal reconstruction via low-rank Hankel matrix
  completion}}, Applied and Computational Harmonic Analysis \textbf{46} (2019),
  no.~1, 94--121.

\bibitem{CaiWeiExploiting18}
Jian-Feng Cai and Ke~Wei, \emph{Exploiting the structure effectively and
  efficiently in low-rank matrix recovery}, Processing, Analyzing and Learning
  of Images, Shapes, and Forms \textbf{19} (2018), 21 pp.

\bibitem{CZ13}
T.~Tony Cai and Anru Zhang, \emph{Sparse representation of a polytope and
  recovery of sparse signals and low-rank matrices}, IEEE transactions on
  information theory \textbf{60} (2013), no.~1, 122--132.

\bibitem{Cai-2013Sparse}
\bysame, \emph{Sparse representation of a polytope and recovery of sparse
  signals and low-rank matrices}, IEEE transactions on information theory
  \textbf{60} (2013), no.~1, 122--132.

\bibitem{Decode}
Emmanuel Cand\`es and Terence Tao, \emph{Decoding by linear programming}, IEEE
  Transactions on Information Theory \textbf{51} (2005), no.~12, 4203--4215.

\bibitem{CT06}
\bysame, \emph{Near optimal signal recovery from random projections and
  universal encoding strategies}, IEEE Transactions on Information Theory
  \textbf{52} (2006), 5406--5425.

\bibitem{CENR11}
Emmanuel~J. Cand\`es, Yonina~C. Eldar, Deanna Needell, and Paige Randall,
  \emph{Compressed sensing with coherent and redundant dictionaries}, Applied
  and Computational Harmonic Analysis \textbf{31} (2011), no.~1, 59--73.

\bibitem{Candes-2015Phase}
Emmanuel~J. Cand\`es, Yonina~C. Eldar, Thomas Strohmer, and Vladislav
  Voroninski, \emph{Phase retrieval via matrix completion}, SIAM Review
  \textbf{57} (2015), no.~2, 225--251.

\bibitem{CRT06}
Emmanuel~J. Cand\`es, Justin~K. Romberg, and Terence Tao, \emph{Stable signal
  recovery from incomplete and inaccurate measurements}, Communications on Pure
  and Applied Mathematics \textbf{59} (2006), no.~8, 1207--1223.

\bibitem{CCL20}
Peter~G. Casazza, Xuemei Chen, and Richard~G. Lynch, \emph{Preserving
  injectivity under subgaussian mappings and its application to compressed
  sensing}, Applied and Computational Harmonic Analysis \textbf{49} (2020),
  no.~2, 451--470.

\bibitem{Chen-2001Atomic}
Scott~Shaobing Chen, David~L. Donoho, and Michael~A. Saunders, \emph{Atomic
  decomposition by basis pursuit}, SIAM review \textbf{43} (2001), no.~1,
  129--159.

\bibitem{C24}
Xuemei Chen, \emph{A unified recovery of structured signals using atomic norm},
  Information and Inference: A Journal of the IMA \textbf{13} (2024), no.~1,
  iaad050.

\bibitem{CWW14}
Xuemei Chen, Haichao Wang, and Rongrong Wang, \emph{A null space analysis of
  the $\ell_1$-synthesis method in dictionary-based compressed sensing},
  Applied and Computational Harmonic Analysis \textbf{37} (2014), no.~3,
  492--515.

\bibitem{Chen-2015Incoherence}
Yudong Chen, \emph{Incoherence-optimal matrix completion}, IEEE Transactions on
  Information Theory \textbf{61} (2015), no.~5, 2909--2923.

\bibitem{chen_chi18}
Yudong Chen and Yuejie Chi, \emph{{Harnessing Structures in Big Data via
  Guaranteed Low-Rank Matrix Estimation: Recent Theory and Fast Algorithms via
  Convex and Nonconvex Optimization}}, {IEEE} Signal Process. Mag. \textbf{35}
  (2018), no.~4, 14--31.

\bibitem{ChenChi-2014}
Yuxin Chen and Yuejie Chi, \emph{{Robust Spectral Compressed Sensing via
  Structured Matrix Completion}}, IEEE Trans. Inf. Theory \textbf{60} (2014),
  no.~10, 6576--6601.

\bibitem{ChiLuChen19}
Yue M.~Lu Chi, Yuejie and Yuxin Chen, \emph{{Nonconvex Optimization Meets
  Low-Rank Matrix Factorization: An Overview}}, IEEE Trans. Signal Process.
  \textbf{67} (2019), no.~20, 5239--5269.

\bibitem{CDTW18}
Abdellah Chkifa, Nick Dexter, Hoang Tran, and Clayton Webster, \emph{Polynomial
  approximation via compressed sensing of high-dimensional functions on lower
  sets}, Mathematics of Computation \textbf{87} (2018), no.~311, 1415--1450.

\bibitem{Clum-2022Derandomized}
Charles Clum and Dustin~G. Mixon, \emph{Derandomized compressed sensing with
  nonuniform guarantees for $\ell_1$ recovery}, Journal of Fourier Analysis and
  Applications \textbf{28} (2022), no.~2, 35.

\bibitem{CDD09}
Albert Cohen, Wolfgang Dahmen, and Ronald DeVore, \emph{Compressed sensing and
  best k-term approximation}, Journal of the American Mathematical Society
  \textbf{22} (2009), 211--231.

\bibitem{DNW}
Mark~A. Davenport, Deanna Needell, and Michael~B. Wakin, \emph{Signal space
  cosamp for sparse recovery with redundant dictionaries}, IEEE Transactions on
  Information Theory \textbf{59} (2012), no.~10, 6820--6829.

\bibitem{Davenport16}
Mark~A. Davenport and Justin Romberg, \emph{An overview of low-rank matrix
  recovery from incomplete observations}, IEEE J. Sel. Topics Signal Process.
  \textbf{10} (2016), 608--622.

\bibitem{DLR18}
S.~{Dirksen}, G.~{Lecu\'{e}}, and H.~{Rauhut}, \emph{On the gap between
  restricted isometry properties and sparse recovery conditions}, IEEE Trans.
  Inf. Theory \textbf{64} (2018), no.~8, 5478--5487.

\bibitem{Dokmanic-2015}
Ivan Dokmani\'c, Reza Parhizkar, Juri Ranieri, and Martin Vetterli,
  \emph{{Euclidean Distance Matrices: Essential theory, Algorithms, and
  Applications}}, IEEE Signal Processing Magazine \textbf{32} (2015), no.~6,
  12--30.

\bibitem{D06}
David Donoho, \emph{{Compressed Sensing}}, IEEE Transactions on Information
  Theory \textbf{52} (2006), no.~4, 1289--1306.

\bibitem{Don18}
David Donoho, \emph{Blackboard to bedside: How high-dimensional geometry is
  transforming the mri industry}, Notices of the AMS \textbf{65} (2018), no.~1,
  40--44.

\bibitem{DE03}
David Donoho and Michael Elad, \emph{Optimally sparse representation in general
  (nonorthogonal) dictionaries via $l_1$ minimization}, Proceedings of the
  National Academy of Science \textbf{100} (2003), no.~5, 2197--202.

\bibitem{Donoho-1992Signal}
David~L. Donoho and Benjamin~F. Logan, \emph{Signal recovery and the large
  sieve}, SIAM Journal on Applied Mathematics \textbf{52} (1992), no.~2,
  577--591.

\bibitem{Dorfler-2022Bridging}
Florian D{\"o}rfler, Jeremy Coulson, and Ivan Markovsky, \emph{Bridging direct
  and indirect data-driven control formulations via regularizations and
  relaxations}, IEEE Transactions on Automatic Control \textbf{68} (2022),
  no.~2, 883--897.

\bibitem{EladAharon-2006ImageDenoising}
Michael Elad and Michal Aharon, \emph{Image denoising via sparse and redundant
  representations over learned dictionaries}, IEEE Transactions on Image
  Processing \textbf{15} (2006), no.~12, 3736--3745.

\bibitem{Elad07}
Michael Elad, Peyman Milanfar, and Ron Rubinstein, \emph{Analysis versus
  synthesis in signal priors}, Inverse Problems \textbf{23} (2007), no.~3,
  947--968.

\bibitem{Fazel_hindi_boyd04}
Maryam Fazel, Haitham~A. Hindi, and Stephen~P. Boyd, \emph{Rank minimization
  and applications in system theory}, Proceedings of the 2004 American Control
  Conference, no.~4, 2004, pp.~3273--3278.

\bibitem{FWHS16}
Christoph Forman, Jens Wetzl, Carmel Hayes, and Michaela Schmidt,
  \emph{Compressed sensing: a paradigm shift in mri}, MAGNETOM Flash (2016),
  19.

\bibitem{F14}
Simon Foucart, \emph{Stability and robustness of $\ell_1$-minimizations with
  weibull matrices and redundant dictionaries}, Linear Algebra and its
  Applications \textbf{441} (2014), 4--21.

\bibitem{Foucart-2016Dictionary}
\bysame, \emph{Dictionary-sparse recovery via thresholding-based algorithms},
  Journal of Fourier Analysis and Applications \textbf{22} (2016), 6--19.

\bibitem{Foucart-2009Lq}
Simon Foucart and Ming-Jun Lai, \emph{Sparsest solutions of underdetermined
  linear systems via $\ell_q$-minimization for $0<q \leq 1$}, Applied and
  Computational Harmonic Analysis \textbf{26} (2009), no.~3, 395--407.

\bibitem{Foucart-2010Gelfand}
Simon Foucart, Alain Pajor, Holger Rauhut, and Tino Ullrich, \emph{The gelfand
  widths of $\ell_p$-balls for $0< p \leq 1$}, Journal of Complexity
  \textbf{26} (2010), no.~6, 629--640.

\bibitem{FR13}
Simon Foucart and Holger Rauhut, \emph{A mathematical introduction to
  compressive sensing}, Birkh\"auser/Springer, New York, 2013.

\bibitem{Fuchs-2022Proof}
Tim Fuchs, David Gross, Peter Jung, Felix Krahmer, Richard Kueng, and Dominik
  St{\"o}ger, \emph{Proof methods for robust low-rank matrix recovery},
  Compressed Sensing in Information Processing, Springer, 2022, pp.~37--75.

\bibitem{Garnaev84}
Andrey~Y. Garnaev and Efim~Davydovich Gluskin, \emph{On widths of the euclidean
  ball}, Journal of Soviet Mathematics \textbf{30} (1984), 200--204.

\bibitem{Giryes-2014Greedy}
Raja Giryes, Sangnam Nam, Michael Elad, R{\'e}mi Gribonval, and Mike~E Davies,
  \emph{Greedy-like algorithms for the cosparse analysis model}, Linear Algebra
  and its Applications \textbf{441} (2014), 22--60.

\bibitem{Giryes-2015Greedy}
Raja Giryes and Deanna Needell, \emph{Greedy signal space methods for
  incoherence and beyond}, Applied and Computational Harmonic Analysis
  \textbf{39} (2015), no.~1, 1--20.

\bibitem{Gluskin-1989Extremal}
Efim~Davydovich Gluskin, \emph{Extremal properties of orthogonal
  parallelepipeds and their applications to the geometry of banach spaces},
  Mathematics of the USSR-Sbornik \textbf{64} (1989), no.~1, 85.

\bibitem{GN03}
R{\'e}mi Gribonval and Morten Nielsen, \emph{Sparse representations in unions
  of bases}, IEEE Transactions on Information Theory \textbf{49} (2003),
  no.~12, 3320--3325.

\bibitem{Gross-2011Recovering}
David Gross, \emph{Recovering low-rank matrices from few coefficients in any
  basis}, IEEE Transactions on Information Theory \textbf{57} (2011), no.~3,
  1548--1566.

\bibitem{Guedon-2022Geometry}
Olivier Gu{\'e}don, Felix Krahmer, Christian K{\"u}mmerle, Shahar Mendelson,
  and Holger Rauhut, \emph{On the geometry of polytopes generated by
  heavy-tailed random vectors}, Communications in Contemporary Mathematics
  \textbf{24} (2022), no.~03, 2150056.

\bibitem{Hastie-2015statistical}
Trevor Hastie, Robert Tibshirani, and Martin Wainwright, \emph{Statistical
  learning with sparsity}, Monographs on Statistics and Applied Probability
  \textbf{143} (2015), no.~143, 8.

\bibitem{H10}
Jarvis Haupt, Waheed~U Bajwa, Gil Raz, and Robert Nowak, \emph{Toeplitz
  compressed sensing matrices with applications to sparse channel estimation},
  IEEE transactions on information theory \textbf{56} (2010), no.~11,
  5862--5875.

\bibitem{Her10}
Felix~J. Herrmann, \emph{Randomized sampling and sparsity: Getting more
  information from fewer samples}, Geophysics \textbf{75} (2010), no.~6,
  WB173--WB187.

\bibitem{Hu-2022LoRA}
Edward~J. Hu, Yelong Shen, Phillip Wallis, Zeyuan Allen-Zhu, Yuanzhi Li, Shean
  Wang, Lu~Wang, and Weizhu Chen, \emph{Lo{RA}: Low-rank adaptation of large
  language models}, International Conference on Learning Representations
  (ICLR), 2022.

\bibitem{Jang-2024LoRA}
Uijeong Jang, Jason~D. Lee, and Ernest~K. Ryu, \emph{Lo{RA} training in the
  {NTK} regime has no spurious local minima}, Forty-first International
  Conference on Machine Learning (ICML), 2024.

\bibitem{Jaspan-2015compressed}
Oren~N Jaspan, Roman Fleysher, and Michael~L Lipton, \emph{Compressed sensing
  mri: a review of the clinical literature}, The British journal of radiology
  \textbf{88} (2015), no.~1056, 20150487.

\bibitem{Jin-2016General}
Kyong~Hwan Jin, Dongwook Lee, and Jong~Chul Ye, \emph{{A general framework for
  compressed sensing and parallel MRI using annihilating filter based low-rank
  Hankel matrix}}, IEEE Transactions on Computational Imaging \textbf{2}
  (2016), no.~4, 480--495.

\bibitem{Jin-2015Annihilating}
Kyong~Hwan Jin and Jong~Chul Ye, \emph{{Annihilating filter-based low-rank
  Hankel matrix approach for image inpainting}}, IEEE Transactions on Image
  Processing \textbf{24} (2015), no.~11, 3498--3511.

\bibitem{KT07}
Boris~S Kashin and Vladimir~N Temlyakov, \emph{A remark on compressed sensing},
  Mathematical notes \textbf{82} (2007), 748--755.

\bibitem{Kol11}
Vladimir Koltchinskii, \emph{{Oracle Inequalities in Empirical Risk
  Minimization and Sparse Recovery Problems: Ecole d'Et{\'e} de
  Probabilit{\'e}s de Saint-Flour XXXVIII-2008}}, vol. 2033, Springer,
  Heidelberg, 2011.

\bibitem{Koltchinskii-2015Bounding}
Vladimir Koltchinskii and Shahar Mendelson, \emph{Bounding the smallest
  singular value of a random matrix without concentration}, Int. Math. Res.
  Not. IRMN \textbf{2015} (2015), no.~23, 12991--13008.

\bibitem{koren_bell_volinsky}
Yehuda Koren, Robert Bell, and Chris Volinsky, \emph{Matrix factorization
  techniques for recommender systems}, Computer \textbf{42} (2009), no.~8,
  30--37.

\bibitem{Krahmer-2015Compressive}
Felix Krahmer, Deanna Needell, and Rachel Ward, \emph{Compressive sensing with
  redundant dictionaries and structured measurements}, SIAM Journal on
  Mathematical Analysis \textbf{47} (2015), no.~6, 4606--4629.

\bibitem{Krahmer-2014structured}
Felix Krahmer and Holger Rauhut, \emph{Structured random measurements in signal
  processing}, GAMM-Mitteilungen \textbf{37} (2014), no.~2, 217--238.

\bibitem{KrahmerWard-2014}
Felix Krahmer and Rachel Ward, \emph{Stable and robust sampling strategies for
  compressive imaging}, IEEE Transactions on Image Processing \textbf{23}
  (2014), no.~2, 612--622.

\bibitem{Kutyniok-2012Sompactly}
Gitta Kutyniok, Jakob Lemvig, and Wang-Q. Lim, \emph{Compactly supported
  shearlets}, Approximation Theory XIII: San Antonio 2010 (New York, NY)
  (Marian Neamtu and Larry Schumaker, eds.), Springer, New York, 2012,
  pp.~163--186.

\bibitem{Levy-1981Reconstruction}
Shlomo Levy and Peter~K Fullagar, \emph{Reconstruction of a sparse spike train
  from a portion of its spectrum and application to high-resolution
  deconvolution}, Geophysics \textbf{46} (1981), no.~9, 1235--1243.

\bibitem{LLMLY12}
Yulong Liu, Shidong Li, Tiebin Mi, Hong Lei, and Weidong Yu, \emph{Performance
  analysis of $\ell_1$-synthesis with coherent frames}, 2012 IEEE International
  Symposium on Information Theory Proceedings, IEEE, 2012, pp.~2042--2046.

\bibitem{Log65}
Benjamin~Franklin Logan, \emph{Properties of high-pass signals}, Ph.D. thesis,
  Columbia University, 1965.

\bibitem{LDP07}
Michael Lustig, David Donoho, and John~M. Pauly, \emph{{Sparse MRI: The
  application of compressed sensing for rapid MR imaging}}, Magnetic Resonance
  in Medicine \textbf{58} (2007), no.~6, 1182--1195.

\bibitem{Mairal-2008Supervised}
Julien Mairal, Jean Ponce, Guillermo Sapiro, Andrew Zisserman, and Francis
  Bach, \emph{Supervised dictionary learning}, Advances in Neural Information
  Processing Systems (NIPS) \textbf{21} (2008), 1033--1040.

\bibitem{Mallat-1999Wavelet}
St{\'e}phane Mallat, \emph{A wavelet tour of signal processing}, Elsevier,
  1999.

\bibitem{MBKW23}
Maximilian M{\"a}rz, Claire Boyer, Jonas Kahn, and Pierre Weiss, \emph{Sampling
  rates for $\ell_1$-synthesis}, Foundations of Computational Mathematics
  \textbf{23} (2023), no.~6, 2089--2150.

\bibitem{Mendelson-2014Remark}
Shahar Mendelson, \emph{A remark on the diameter of random sections of convex
  bodies}, Geometric Aspects of Functional Analysis: Israel Seminar (GAFA)
  2011-2013, Springer, 2014, pp.~395--404.

\bibitem{MendelsonLearning15}
\bysame, \emph{Learning without concentration}, J. ACM \textbf{62} (2015),
  no.~3, 21:1--21:25.

\bibitem{ML17}
Shahar Mendelson and Guillaume Lecu\'{e}, \emph{Sparse recovery under weak
  moment assumptions}, J. Eur. Math. Soc. \textbf{19} (2017), no.~3, 881--904.

\bibitem{MPT08}
Shahar Mendelson, Alain Pajor, and Nicole Tomczak-Jaegermann, \emph{Uniform
  uncertainty principle for bernoulli and subgaussian ensembles}, Constructive
  Approximation \textbf{28} (2008), 277--289.

\bibitem{MRW18}
Shahar Mendelson, Holger Rauhut, and Rachel Ward, \emph{Improved bounds for
  sparse recovery from subsampled random convolutions}, Ann. Appl. Probab.
  \textbf{28} (2018), no.~6, 3491--3527.

\bibitem{Monajemi-2013Deterministic}
Hatef Monajemi, Sina Jafarpour, Matan Gavish, Stat 330/CME~362 Collaboration,
  David~L. Donoho, Sivaram Ambikasaran, Sergio Bacallado, Dinesh Bharadia,
  Yuxin Chen, Young Choi, et~al., \emph{Deterministic matrices matching the
  compressed sensing phase transitions of gaussian random matrices},
  Proceedings of the National Academy of Sciences \textbf{110} (2013), no.~4,
  1181--1186.

\bibitem{Papyan-2018Theoretical}
Vardan Papyan, Yaniv Romano, Jeremias Sulam, and Michael Elad,
  \emph{Theoretical foundations of deep learning via sparse representations: A
  multilayer sparse model and its connection to convolutional neural networks},
  IEEE Signal Processing Magazine \textbf{35} (2018), no.~4, 72--89.

\bibitem{PelegElad-2013}
Tomer Peleg and Michael Elad, \emph{Performance guarantees of the thresholding
  algorithm for the cosparse analysis model}, IEEE Transactions on Information
  Theory \textbf{59} (2013), no.~3, 1832--1845.

\bibitem{RRT12}
Holger Rauhut, Justin Romberg, and Joel~A Tropp, \emph{Restricted isometries
  for partial random circulant matrices}, Applied and Computational Harmonic
  Analysis \textbf{32} (2012), no.~2, 242--254.

\bibitem{DMP1}
Holger Rauhut, Karin Schnass, and Pierre Vandergheynst, \emph{Compressed
  sensing and redundant dictionaries}, IEEE Transactions on Information Theory
  \textbf{54} (2008), no.~5, 2210--2219.

\bibitem{Ravishankar2019image}
Saiprasad Ravishankar, Jong~Chul Ye, and Jeffrey~A Fessler, \emph{Image
  reconstruction: From sparsity to data-adaptive methods and machine learning},
  Proceedings of the IEEE \textbf{108} (2019), no.~1, 86--109.

\bibitem{recht2013parallel}
Benjamin Recht and Christopher R{\'e}, \emph{Parallel stochastic gradient
  algorithms for large-scale matrix completion}, Mathematical Programming
  Computation \textbf{5} (2013), no.~2, 201--226.

\bibitem{Rish-2014Sparse}
Irina Rish and Genady Grabarnik, \emph{Sparse modeling: Theory, algorithms, and
  applications}, CRC Press, 2014.

\bibitem{RV08}
Mark Rudelson and Roman Vershynin, \emph{On sparse reconstruction from
  {F}ourier and {G}aussian measurements}, Communications on Pure and Applied
  Mathematics \textbf{61} (2008), 1025--1045.

\bibitem{Saab-2010Sparse}
Rayan Saab and {\"O}zg{\"u}r Y{\i}lmaz, \emph{Sparse recovery by non-convex
  optimization--instance optimality}, Applied and Computational Harmonic
  Analysis \textbf{29} (2010), no.~1, 30--48.

\bibitem{Santosa86}
Fadil Santosa and William~W. Symes, \emph{Linear inversion of band-limited
  reflection seismograms}, SIAM Journal on Scientific and Statistical Computing
  \textbf{7} (1986), no.~4, 1307--1330.

\bibitem{Sengupta-2017New}
Soumyadip Sengupta, Tal Amir, Meirav Galun, Tom Goldstein, David~W Jacobs, Amit
  Singer, and Ronen Basri, \emph{A new rank constraint on multi-view
  fundamental matrices, and its application to camera location recovery},
  Proceedings of the IEEE Conference on Computer Vision and Pattern Recognition
  (CVPR), 2017, pp.~4798--4806.

\bibitem{Shechtman-2015Phase}
Yoav Shechtman, Yonina~C Eldar, Oren Cohen, Henry~Nicholas Chapman, Jianwei
  Miao, and Mordechai Segev, \emph{Phase retrieval with application to optical
  imaging: A contemporary overview}, IEEE Signal Processing Magazine
  \textbf{32} (2015), no.~3, 87--109.

\bibitem{Talagrand-2021Upper}
Michel Talagrand, \emph{Upper and lower bounds for stochastic processes:
  Decomposition theorems}, Ergebnisse der Mathematik und ihrer Grenzgebiete. 3.
  Folge / A Series of Modern Surveys in Mathematics, vol.~60, Springer Science
  \& Business Media, 2021.

\bibitem{Tasissa-2018EDG}
Abiy Tasissa and Rongjie Lai, \emph{Exact reconstruction of euclidean distance
  geometry problem using low-rank matrix completion}, IEEE Transactions on
  Information Theory \textbf{65} (2018), no.~5, 3124--3144.

\bibitem{Thesing-2021Non}
Laura Thesing and Anders~Christian Hansen, \emph{Non-uniform recovery
  guarantees for binary measurements and infinite-dimensional compressed
  sensing}, Journal of Fourier Analysis and Applications \textbf{27} (2021),
  no.~2, 14.

\bibitem{Tibshirani-1996Lasso}
Robert Tibshirani, \emph{Regression shrinkage and selection via the lasso},
  Journal of the Royal Statistical Society: Series B (Methodological)
  \textbf{58} (2018), no.~1, 267--288.

\bibitem{Tillmann-2013Computational}
Andreas~M. Tillmann and Marc~E. Pfetsch, \emph{The computational complexity of
  the restricted isometry property, the nullspace property, and related
  concepts in compressed sensing}, IEEE Transactions on Information Theory
  \textbf{60} (2013), no.~2, 1248--1259.

\bibitem{Tirer-2020GeneralizingCoSaMP}
Tom Tirer and Raja Giryes, \emph{Generalizing cosamp to signals from a union of
  low dimensional linear subspaces}, Applied and Computational Harmonic
  Analysis \textbf{49} (2020), no.~1, 99--122.

\bibitem{T04}
Joel~A. Tropp, \emph{Greed is good: Algorithmic results for sparse
  approximation}, IEEE Transactions on Information theory \textbf{50} (2004),
  no.~10, 2231--2242.

\bibitem{T15}
\bysame, \emph{Convex recovery of a structured signal from independent random
  linear measurements}, Sampling Theory, a Renaissance: Compressive Sensing and
  Other Developments (2015), 67--101.

\bibitem{Ver10}
Roman Vershynin, \emph{Introduction to the non-asymptotic analysis of random
  matrices}, Compressed Sensing, Theory and Applications (Y.~Eldar and
  G.~Kutyniok, eds.), Cambridge University Press, Cambridge, 2012,
  pp.~210--268.

\bibitem{Ver18:High-Dimensional-Probability}
Roman Vershynin, \emph{High-dimensional probability: An introduction with
  applications in data science}, Cambridge Series in Statistical and
  Probabilistic Mathematics, vol.~47, Cambridge University Press, Cambridge,
  2018.

\bibitem{Vidyasagar-2019CS}
Mathukumalli Vidyasagar, \emph{An introduction to compressed sensing}, Society
  for Industrial and Applied Mathematics, Philadelphia, PA, 2019.

\bibitem{Wipf-2010Iterative}
David Wipf and Srikantan Nagarajan, \emph{Iterative reweighted $\ell_1$ and
  $\ell_2$ methods for finding sparse solutions}, IEEE Journal of Selected
  Topics in Signal Processing \textbf{4} (2010), no.~2, 317--329.

\bibitem{W10}
Przemyslaw Wojtaszczyk, \emph{Stability and instance optimality for gaussian
  measurements in compressed sensing}, Foundations of Computational Mathematics
  \textbf{10} (2010), 1--13.

\bibitem{Yi-2020Necessary}
Jirong Yi and Weiyu Xu, \emph{Necessary and sufficient null space condition for
  nuclear norm minimization in low-rank matrix recovery}, IEEE Transactions on
  Information Theory \textbf{66} (2020), no.~10, 6597--6604.

\bibitem{Zeng-2024ExpressiveLoRA}
Yuchen Zeng and Kangwook Lee, \emph{The expressive power of low-rank
  adaptation}, The Twelfth International Conference on Learning Representations
  (ICLR), 2024.

\bibitem{Zhao-2018Sparse}
Yun-Bin Zhao, \emph{Sparse optimization theory and methods}, 1st ed., CRC
  Press, 2018.

\bibitem{Zheng2017-LpL1}
Le~Zheng, Arian Maleki, Haolei Weng, Xiaodong Wang, and Teng Long, \emph{Does
  $\ell_{p}$-minimization outperform $\ell_{1}$-minimization?}, IEEE
  Transactions on Information Theory \textbf{63} (2017), no.~11, 6896--6935.

\end{thebibliography}
\end{document}